\documentclass[11pt]{amsart}

\usepackage[english]{babel}
\usepackage[utf8]{inputenc}
\usepackage[T1]{fontenc}

\usepackage{mathtools}
\usepackage{amsmath,amsfonts,amsthm}
\usepackage{physics}
\usepackage{graphicx}
\usepackage{subfig}
\usepackage{enumitem}
\usepackage{braket}
\usepackage{hyperref}

\usepackage[widespace]{fourier}
\usepackage{subfig}
\usepackage{wrapfig}
\usepackage{xcolor}
\usepackage{subfig}
\usepackage{amscd}

\usepackage{etoolbox} % for '\AtBeginEnvironment' macro
\AtBeginEnvironment{pmatrix}{\everymath{\displaystyle}}
\AtBeginEnvironment{bmatrix}{\everymath{\displaystyle}}

\DeclareMathOperator{\Id}{Id}
\DeclareMathOperator{\ud}{d}

\DeclareMathOperator{\SO}{SO}

\DeclarePairedDelimiter{\poisson}{\{}{\}}
\newcommand{\Alg}{\mathfrak{A}}

\numberwithin{equation}{section}

\newcommand{\Pj}{\mathbb{CP}}
\newcommand{\Sph}{\mathbb{S}}
\newcommand{\RR}{\mathbb{R}}
\newcommand{\ZZ}{\mathbb{Z}}
\newcommand{\N}{\mathbb{N}}
\newcommand{\Cp}{\mathbb{C}}
\newcommand{\dP}{\mathrm{d}P}

\newcommand{\Sl}{\mathfrak{sl}}
\newcommand{\vx}{\vec{x}}

\renewcommand{\vec}[1]{\mathbf{#1}}

\renewcommand{\epsilon}{\varepsilon}
\renewcommand{\imath}{\mathrm{i}}

\newcommand{\valpha}{\boldsymbol{\alpha}}
\newcommand{\vbeta}{\boldsymbol{\beta}}

\allowdisplaybreaks
\makeatletter 
\renewcommand{\pdv}[2]{\begingroup 
\@tempswafalse\toks@={}\count@=\z@ 
\@for\next:=#2\do 
{\expandafter\check@var\next\@nil
 \advance\count@\der@exp 
 \if@tempswa 
   \toks@=\expandafter{\the\toks@\,}% 
 \else 
   \@tempswatrue 
 \fi 
 \toks@=\expandafter{\the\expandafter\toks@\expandafter\partial\der@var}}% 
\frac{\partial\ifnum\count@=\@ne\else^{\number\count@}\fi#1}{\the\toks@}% 
\endgroup} 
\def\check@var{\@ifstar{\mult@var}{\one@var}} 
\def\mult@var#1#2\@nil{\def\der@var{#2^{#1}}\def\der@exp{#1}} 
\def\one@var#1\@nil{\def\der@var{#1}\chardef\der@exp\@ne} 
\makeatother

\theoremstyle{plain}
\newtheorem{theorem}{Theorem}[section]
\newtheorem{proposition}[theorem]{Proposition}

\newtheorem{lemma}[theorem]{Lemma}
\theoremstyle{definition}
\newtheorem{definition}{Definition}[section]

\theoremstyle{remark}
\newtheorem{remark}{Remark}[section]
\theoremstyle{remark}
\newtheorem{example}{Example}

\newtheorem*{conjecture*}{Conjecture}

\usepackage[backend=biber,sorting=nyt,maxnames=99,giveninits=true,doi=false,citestyle=numeric-comp]{biblatex}
\begin{document}

\title{Coalgebra symmetry for discrete systems}

%\author{G. Gubbiotti \and D. Latini \and B. K. Tapley}

\author[G. Gubbiotti]{Giorgio Gubbiotti}
\address{Scuola Internazionale Superiore di Studi Avanzati, Via Bonomea 265, 34136 Trieste, Italy}
\curraddr{Dipartimento di Matematica ``Federigo Enriques'', Universit\`a degli Studi di Milano, Via C. Saldini 50, 20133 Milano, Italy}
\email{giorgio.gubbiotti@unimi.it}
\author[D. Latini]{Danilo Latini}
\address{School of Mathematics and Physics, The University of Queensland, Brisbane, QLD 4072, Australia}
\email{d.latini@uq.edu.au}
\author[B. K. Tapley]{Benjamin K. Tapley}
\address{Department of Mathematical Sciences, NTNU, 7491 Trondheim, Norway}
\email{benjamin.tapley@ntnu.no}

\maketitle

\begin{abstract}
    In this paper we introduce the notion of coalgebra symmetry for
    discrete systems. With this concept we prove that all discrete
    radially symmetric systems in standard form are quasi-integrable
    and that all variational discrete quasi-radially symmetric systems
    in standard form are Poincar\'e--Lyapunov--Nekhoroshev maps of
    order $N-2$, where $N$ are the degrees of freedom of the system.
    We also discuss the integrability properties of several vector
    systems which are generalisations of well-known one degree of freedom
    discrete integrable systems, including two $N$ degrees of freedom
    autonomous discrete Painlev\'e I equations and an $N$ degrees of
    freedom McMillan map.
\end{abstract}

\section{Introduction}

In this paper we show how to use the coalgebra symmetry approach to build 
invariants for discrete systems. In particular, we consider systems of 
second order difference equations:
\begin{equation}
    \vx_{n+1} = \vec{F}\left( \vx_{n},\vx_{n-1}\right),
    \label{eq:map2ndord}
\end{equation}
where the unknown is a sequence of vectors
$\left\{\vec{x}_{n}\right\}_{n\in\ZZ}$ in $\RR^{N}$ with $N\in\N$.
Like in the continuous case the coalgebra symmetry approach allows us to
extend invariants from symplectic systems with one degree of freedom
to symplectic systems with $N$ degrees of freedom, see for instance the
review \cite{Ballesteros_et_al2009}.

Besides the general definition, we present some general results on the
integrability through the coalgebra symmetry approach of some classes of
$N$ degrees of freedom analog of the systems in \emph{standard form}
\cite{Suris1989}:
\begin{equation}
    x_{n+1} + x_{n-1} = F\left( x_{n} \right),
    \label{eq:mst1d}
\end{equation}
where $\left\{ x_{n} \right\}_{n\in\ZZ}$ is the unknown function.
Integrable examples include a generalisation of the celebrated McMillan
equation \cite{McMillan1971}, one of the first integrable discrete
systems ever discovered, and of the autonomous discrete Painleve I
equation. We underline that a two-degrees of freedom generalisation
of the McMillan equation presented already in \cite{McLachlan1993} by
a brute-force computation, and later understood in terms of discrete
Garnier systems \cite{Suris1994Garnier}. On the other hand, up to our
knowledge, the generalisations of the autonomous Painlev\'e I equation
we present are new.

%In this paper we put such a system in the new general framework of 
%the coalgebra symmetry for discrete systems, 
%and we present a new quasi-integrable generalisation of such a map.

The extension of the coalgebra approach to the discrete setting we present
is a result that paves the way to a systematic study of
the integrability conditions for discrete systems with more than
one degree of freedom.
Indeed, while almost all discrete systems with one degree of freedom
fall into the class of the aforementioned QRT maps,
with some notable exceptions
\cite{VialletRamaniGrammaticos2004,Duistermaat2011book, RJ15,CarsteaTakenawa2012},
there is no general analog description for systems with more degrees
of freedom.
Partial classification of systems in more than one degree of freedom
has been given in 
\cite{GJTV_class,CapelSahadevan2001,Gubbiotti_lagr_add,McLachlan1993},
usually with the additional condition that the system possesses
invariants of a given form, or additional structures such as symplectic
structures.

So, in this paper we present an efficient way to build discrete systems
with $N$ degrees of freedom from those with one degree of freedom,
which is something that is completely missing in the discrete case.
This gives a new evidence on how techniques developed in the continuum
setting can be extended proficiently to the discrete case. We also
show that the access to a very efficient integrability test such as the
algebraic entropy \cite{BellonViallet1999}, makes easier to predict the
behaviour of the obtained $N$-degrees of freedom systems.

The plan of the paper is the following: in Section \ref{sec:preliminaries}
we recall some general result on the integrability of discrete
systems.  In Section \ref{sss:coalgebra} we give the definition
of coalgebra symmetry for discrete Poisson maps. The great part of
the construction follows the ideas used in the continuous setting,
see \cite{Ballesteros_et_al2009}.  However, the definition we state
is slightly different from the one used in the continuous setting.
In Section \ref{sec:general} we present some general results on two
classes of discrete systems, namely the radially symmetric and the
quasi-radially symmetric, show their coalgebraic interpretation, and
present explicit examples of it. In Section \ref{sec:dPI} we discuss
two non-linear examples which generalise a well-known integrable system,
the autonomous discrete Painlev\'e I equation. In Section \ref{sec:dPII}
we discuss an $N$ degrees of freedom generalisation of the McMillan
map. We show that this generalisation is not integrable, but it is
only quasi-integrable, while a particular case is quasi-maximally
superintegrable.  Finally, in Section \ref{sec:concl} we give some
conclusions and an outlook on future researches and developments.

\section{Discrete systems and integrability}
\label{sec:preliminaries}

In this section we give a preliminary introduction on the concepts of
integrability for finite-dimensional discrete systems we are going to
use throughout the paper. Differently from finite-dimensional continuous
integrable systems whose history can be traced back to the times of
Liouville \cite{Liouville1855}, the integrability of their discrete
counterpart is a much more recent subject.  The history of such a subject
can be traced back from the seminal paper of McMillan \cite{McMillan1971},
where the eponymous map was introduced, and had a great propulsion after
the introduction of the celebrated QRT maps \cite{QRT1988,QRT1989}.
In particular, the theory of symplectic integrable maps was developed
in \cite{Veselov1991,Bruschietal1991,Maeda1987}.  For a complete
overview on the subject we refer to the mentioned papers, the review
\cite{HietarintaBook}, the book \cite{HietarintaJoshiNijhoff2016},
and the introductory material in the thesis \cite{TranPhDThesis}.

\subsection{Invariants and integrability}
\label{sss:integrability}

Before considering the vector difference equations in the form \eqref{eq:map2ndord}
we consider the general case:
let $\{\vec{z}_{n}\}\subset\RR^{M}$ be a sequence of vectors, solving the
\emph{first-order} vector difference equation
\begin{equation}
    \vec{z}_{n+1} = \vec{K}\left( \vec{z}_{n} \right),
    \label{eq:zeq}
\end{equation}
where $\vec{K}=\vec{K}(\vec{z})$ is a locally analytic function of
its arguments.
This is what we call an $M$-dimensional difference equation.
An \emph{invariant} for such an equation is a locally analytic function
$I=I(\vec{z})$ such that $I(\vec{z}_{n+1})=I(\vec{z}_{n})$.
In such general case we give the following definition:
\begin{definition}
    The first-order $M$-dimensional difference equation is al\-ge\-bra\-i\-cal\-ly 
    integrable if \emph{it admits $M-1$ functionally independent invariants}.
    \label{def:invint}
\end{definition}

In general it is quite difficult to find invariant for difference
equations. When the equation is rational and invertible with a
rational inverse there is an efficient algorithm we recall in Appendix
\ref{app:findinv}.

Definition \ref{def:invint} is very general, and works for arbitrary maps.
If some additional structure is present, the number of invariants needed
for integrability can be lowered. Before proceeding any further, we note
that any first order difference equation \eqref{eq:zeq} is equivalent
to iterate a map of the following form:
\begin{equation}
    T \colon \vec{z}\in\RR^{M} \mapsto \vec{K}(\vec{z})\in\RR^{M}.
    \label{eq:map}
\end{equation}
We call this form the \emph{map form} of a first-order difference equation.
Since the iteration of the map is equivalent to go from $n$ to $n+1$ we use
the shorthand notation $T\colon n\mapsto n+1$, to denote the map form of
a difference equation.

A special, but relevant case is the one of Poisson maps:
\begin{definition}
    Let us assume we are given an $M$-dimensional difference equation
    \eqref{eq:zeq} and its map form \eqref{eq:map}.
    Then:
    \begin{enumerate}[label=\textbf{\ref{def:poisson}.\arabic*.},leftmargin=*,widest=a]
        \item A \emph{Poisson structure of rank $2r$} is a skew-symmetric matrix 
            $J=J\left( \vec{z}_{n} \right)$ of constant rank $2r$ such that the 
            \emph{Jacobi identity holds}:
            \begin{equation}
                \sum_{l=1}^{n}\left(J_{li}\frac{\partial J_{jk}}{\partial z_{n,l-1}}
                    +J_{lj}\frac{\partial J_{ki}}{\partial z_{n,l-1}}
                +J_{lk}\frac{\partial J_{ij}}{\partial z_{n,l-1}}\right)=0,
                    \quad
                    \forall i,j,k.
                \label{eq:JacoIden}
            \end{equation}
        \item Given a Poisson structure $J=J(\vec{z}_{n})$ we define its
            associated \emph{Poisson bracket} as:
            \begin{equation}
            \label{eq:Poisson_def}
            \{f, g\}=\nabla f   J \left( \vec{z}_{n} \right)  \nabla g^{T},
            \end{equation}
            where $f$, $g$ are locally analytic functions on $\RR^{M}$, and
            $\nabla$ denotes the gradient operator.
        \item Two locally analytic functions $f$ and $g$ are said to be 
            \emph{in involution} with respect to a Poisson structure 
            $J=J\left( \vec{z}_{n} \right)$ if $\left\{ f,g \right\}=0$.
        \item An $M$-dimensional difference equation is a \emph{Poisson map} 
            if it preserves the Poisson structure $J\left( \vec{z}_{n} \right)$, i.e.:
            \begin{equation}
                \label{eq:Poisson1}
                \frac{\partial \vec{z}_{n+1}}{\partial\vec{z}_{n}}  
                J({\vec{z}_{n}})  
                \left(\frac{\partial \vec{z}_{n+1}}{\partial\vec{z}_{n}}\right)^{T}
                =J({\vec{z}_{n+1}}),
            \end{equation}
            where $\partial \vec{z}_{n+1}/\partial\vec{z}_{n}$
            is the  Jacobian matrix of $\vec{z}_{n+1}$. 
    \end{enumerate} 
    \label{def:poisson}
\end{definition}

\begin{remark}
    We can easily see that $\{z_{n,i-1},z_{n,j-1}\}=J_{ij}(\vec{z}_{n})$.
    Since this completely specifies the Poisson structure, it is
    usual to give it in terms of the commutation relations of the
    dependent variables $\vec{z}_{n}$.
    \label{rem:Jij}
\end{remark}

Then we have the following characterisation of integrability for
Poisson maps:
\begin{definition}[\cite{Veselov1991,Bruschietal1991,Maeda1987}]
    An $M$-dimensional Poisson map $T\colon n \mapsto n+1$ with respect
    to a Poisson structure of rank $2r$ is \emph{Liou\-vil\-le--Poisson
    integrable} if it possesses $M-r$ functionally independent invariants
    in involution with respect to the associated Poisson bracket.
    \label{def:lpint}
\end{definition}

\begin{remark}
    The rank of a Poisson structure is such that
    $1\leq r \leq \lfloor M/2\rfloor$.
    So we can distinguish two extremal cases:
    \begin{itemize}
        \item Minimal rank $r=1$: the number of invariants needed for 
            Liouville--Poisson integrability is maximal, that is $M-1$.
            In this case Liouville--Poisson integrability is equivalent 
            to algebraic integrability, see definition \ref{def:invint}.
        \item Maximal rank $r=\lfloor M/2\rfloor$: the number of invariants
            needed for Liouville--Poisson integrability is minimal, that is 
            $M-\lfloor M/2\rfloor$.
    \end{itemize}
    A relevant case is when the rank is maximal and 
    the dimension is even, that is $M=2r$.
    In such case the Poisson structure is invertible; defining
    $\Omega =J^{-1}$ 
    we call such a matrix a \emph{ symplectic structure}, and a
    map preserving it is a \emph{symplectic map}.
    In the full rank case the number of invariants needed is exactly
    half of the dimension of the base space $M/2$.
    In such a case the number $N:=M/2$ is called the number of 
    \emph{degrees of freedom} of the system.
    \label{rem:symplectic}
\end{remark}

The previous remark highlight a very special case of Liouville--Poisson
integrability, which we state as a separate definition:
\begin{definition}[\cite{Veselov1991,Bruschietal1991,Maeda1987}]
    An $2N$-dimensional ($N$ degrees of freedom) symplectic map $T\colon n \mapsto n+1$ is
    \emph{Liou\-vil\-le integrable} if it possesses $N$ functionally
    independent invariants in involution with respect to the associated Poisson
    bracket.
    \label{def:lint}
\end{definition}

Furthermore, we give some additional definition to cover the cases when
there exist more or less invariants than the number given in Definitions
\ref{def:invint} and \ref{def:lpint}.

\begin{definition}
    Suppose we are given a Poisson map $T\colon n\mapsto n+1$ \eqref{eq:map}, 
    with associated Poisson structure of rank $2r$.
    Then:
    \begin{enumerate}[label=\textbf{\ref{def:pintegrability}.\arabic*.},leftmargin=*,widest=a]
        \item If $T\colon n\mapsto n+1$ is Liouville--Poisson integrable,
            and possesses $k$ additional functionally independent invariants,
            with $0< k< r$, 
            then it is said to be \emph{superintegrable}.
            Moreover:
            \begin{enumerate}[label=\textbf{\ref{def:pintegrability}.\arabic{enumi}.\alph*.},leftmargin=*,widest=a]
                \item if $k=1$ the map is said to be \emph{minimally superintegrable},
                \item if $k=r-2$ the map is said to be \emph{quasi-maximally superintegrable},
                \item if $k=r-1$ the map is said to be \emph{maximally superintegrable}.
            \end{enumerate}
        \item If $T\colon n\mapsto n+1$ possesses $\kappa$ commuting 
            functionally independent invariants with $0< \kappa< M-r$, 
            then it is said to be \emph{Poincar\'e--Lyapunov--Nekhoroshev (PLN) map of order $\kappa$}.
            Moreover:
            \begin{enumerate}[label=\textbf{\ref{def:pintegrability}.\arabic{enumi}.\alph*.},leftmargin=*,widest=a]
                \item if $\kappa=1$ the map is said to be \emph{Poincar\'e--Lyapunov (PL) map},
                \item if $\kappa=M-r-1$ the map is said to be \emph{quasi-integrable},
            \end{enumerate}
    \end{enumerate}
    \label{def:pintegrability}
\end{definition}

\begin{remark}
    We make the following observations:
    \begin{itemize}
        \item In the literature on superintegrable systems, see for instance
            \cite{MillerPostWinternitz2013R}, it is often omitted
            in the definition of superingrability the requirement for the system
            to be Liouville--Poisson integrable.
            In fact there are known examples in the literature where it is possible
            to find a huge number of invariants, but not a full set of commuting
            ones, see \cite{BallesterosHerranz2001}.
            So, to avoid the situation where a superintegrable system might not be a
            Liouville--Poisson integrable systems, we prefer to stick to a definition
            where superintegrability requires Liouville--Poisson integrability.
            We will see an example of discrete system with such a property in
            Section \ref{sec:dPII}.
        \item In the symplectic case ($M=2N$) a symplectic map is:
            \begin{itemize}
                \item superintegrable when it is Liouville integrable and 
                    possesses $N<k<2N$ invariants,
                \item quasi-maximally superintegrable when it is superintegrable
                    and possesses $2N-2$ invariants, 
                \item maximally superintegrable when it is superintegrable
                    and possesses $2N-1$ invariants.
                \item quasi-integrable when it possesses $N-1$ commuting invariants.
            \end{itemize}
    \end{itemize}
    \label{rem:minim}
\end{remark}

\subsection{Variational discrete systems and their integrability}

Given a $M$-di\-men\-sion\-al difference equation \eqref{eq:zeq},
it is not easy to decide if there exists a compatible Poisson or
symplectic structure.
Some partial results were given in
\cite{ByrnesHaggarQuispel1999}.  An important particular case is when
the equation is \emph{variational},
that is when the system arises from a discrete Lagrangian,
see \cite{Logan1973,Veselov1991,Bruschietal1991,TranvanderKampQusipel2016,Gubbiotti_dcov}
for further details.
In particular we also mention that the relationship of variational structures
with Liouville integrability was used in \cite{GJTV_class,Gubbiotti_lagr_add}
to prove integrability of some four-dimensional difference equations. 

Let us stick to our case of interest: a system of second-order difference
equations of the form \eqref{eq:map2ndord} is \emph{variational} if there
exist a function
\begin{equation}
    L = L_{n}(\vx_{n+1},\vx_{n}),
    \label{eq:dLagr}
\end{equation}
called the \emph{discrete Lagrangian},
such that the system itself is equivalent to the \emph{Euler-Lagrange} equations:
\begin{equation}
    \nabla_{\vx_{n}}
    \left[ L_{n}(\vx_{n+1},\vx_{n})+ L_{n-1}(\vx_{n},\vx_{n-1})\right]
    =0.
    \label{eq:elgen}
\end{equation}
Note that in general the Euler--Lagrange equations are a multi-valued
function in $\vx_{n+1}$ and $\vx_{n-1}$, see \cite{HietarintaJoshiNijhoff2016}.
Then we have the following result:
\begin{lemma}[\cite{Bruschietal1991,HietarintaJoshiNijhoff2016}]
    If the discrete Lagrangian $L$ does not depend explicitly on $n$,
    then the Euler--Lagrange equations \eqref{eq:elgen} leave invariant the 
    following symplectic structure:
    \begin{equation}
        \Omega(\vx_{n},\vx_{n-1}) =
        \begin{pmatrix}
            \mathbb{O}_{N} & \Lambda_{N} (\vx_{n},\vx_{n-1})
            \\
            -\Lambda_{N}(\vx_{n},\vx_{n-1}) & \mathbb{O}_{N}
        \end{pmatrix}
        \setlength{\delimitershortfall}{0pt}
        \label{eq:OmegaVar}
    \end{equation}
    where $\mathbb{O}_{N}$ is the zero $N \times N$ matrix and:
    \begin{equation}
        \Lambda_{N} (\vx_{n},\vx_{n-1}) =
        \begin{pmatrix}
            \dfrac{\partial^2 L}{\partial x_{n,1}\partial x_{n-1,1}}
            &
            \dots
            &
            \dfrac{\partial^2 L}{\partial x_{n,1}\partial x_{n-1,N}}
            \\[2ex]
            \vdots & &   \vdots
            \\[2ex]
            \dfrac{\partial^2 L}{\partial x_{n,N}\partial x_{n-1,1}}
            &
            \dots
            &
            \dfrac{\partial^2 L}{\partial x_{n,N}\partial x_{n-1,N}}
        \end{pmatrix}.
        \label{eq:Lambdadef}
    \end{equation}
    \label{lem:poisson}
\end{lemma}

For generalisations and proof of such a result we refer to
\cite{TranPhDThesis}.

\section{Coalgebra symmetry for discrete systems}
\label{sss:coalgebra}

The coalgebra symmetry approach to (super)integrable systems is an
approach that allows to build a (classical) Liouville integrable system
on the tensor product of $N$ copies of a given Poisson algebra $\Alg$.
The concept of coalgebra originated in Quantum Group theory and it was not
until the papers \cite{Ballesteros_et_al_1996,BallesterosRagnisco1998}
that its application to (classical) Liouville integrable was devised.
We note that besides these seminal applications the coalgebra symmetry
approach has been extended to many other cases, and many integrable
and superintegrable systems have been understood within this algebraic
framework. For a comprehensive review containing several explicit
examples (for deformed coalgebras also), some generalisations of
the method (such as comodule algebras \cite{Ballesteros_et_al2002}
and Loop coproducts \cite{Musso2010loop}) and many references on the
topic we refer the reader to \cite{Ballesteros_et_al2009}. 
Additional examples, where other physical applications can be found 
(also in the quantum mechanical setting), involve superintegrable 
systems defined on non-Euclidean spaces
\cite{BallesterosHerranz2007,Ballesteros_et_al2008PhysD,Ballesteros_et_al2009AnnPhys,
BallesterosHerranz2009, Ballesteros_et_al2011AnnPhys, Riglioni2013,
PostRiglioni2015}, models with spin-orbital interactions
\cite{Riglioni_et_al2014}, discrete quantum mechanical systems
\cite{LatiniRiglioni2016} and superintegrable systems related to the
generalised Racah algebra $R(n)$  \cite{DeBie_et_al2021, Latini2019,
Latini_et_al2020embedding, Latini_et_al2021}.

Roughly speaking, the general idea behind this algebraic approach is to
reinterpret one degrees of freedom dynamical Hamiltonian systems as the
images, under a given symplectic realisation, of some (smooth) functions
of the generators of a given Poisson (co)algebra $(\Alg, \Delta)$. Then,
by using the (not necessarily primitive) coproduct, the method consists in
extending the one degrees of freedom system, originally defined on $\Alg$,
to a higher degrees of freedom one defined on the tensor product of $N$
copies of $\Alg$. The main point resides in the fact that the obtained
higher degrees of freedom system is endowed with constants of motion
arising as the images of the coalgebra Casimir invariants through the
application of the so-called $m$th coproduct maps.  In this section, we
will recall the definition of coalgebra and then adapt it to the study of
(super)integrable discrete systems.

\subsection{Definition of coalgebra and coproduct}
We state the following main definition:

\begin{definition}[\cite{ChariPressley1994Book,Drinfeld1987}]
    A \emph{coalgebra} is a pair of objects $(\Alg,\Delta)$ where 
    $\Alg$ is a unital, associative algebra and
    $\Delta\colon \Alg \rightarrow \Alg\otimes \Alg$
     is a \emph{coassociative} map, that is:
    \begin{equation} 
        (\Delta \otimes \Id) \circ \Delta=(\Id \otimes \Delta) \circ \Delta.
        \label{eq:coass}
    \end{equation}
    meaning that the following diagram is commutative:
    \begin{equation}
        \begin{CD}    
            \Alg @>\Delta>> \Alg\otimes\Alg
            \\
            @V{\Delta}VV @VV{\Delta\otimes\Id}V 
            \\
            \Alg\otimes\Alg @>>{\Id\otimes\Delta}>
            \Alg\otimes\Alg\otimes\Alg
        \end{CD}
        \label{eq:commdiag}
    \end{equation}
    and is an algebra homomorphism from $\mathfrak{A}$ to  
    $\mathfrak{A} \otimes \mathfrak{A}$, \emph{i.e.}
    $\Delta (X \cdot  Y) = \Delta (X)  \cdot \Delta (Y)$,
    for all $X, Y \in \mathfrak{A}$.
    The map $\Delta$ is called the \emph{coproduct map}.
    \label{def:coalgebra}
\end{definition}

%\begin{remark}
%    The coassociativity property can be expressed by saying that the 
%    This shows, for instance, that the coproduct map $\Delta$
%    provides a `two-fold way' to define elements of the tensor product
%    $\Alg\otimes \Alg\otimes \Alg$.
%    As we shall explain later, this is connected with \emph{superintegrability}.
%    \label{rem:commdiag}
%\end{remark}

As for Definition \ref{def:coalgebra} the base algebra $\Alg$ of a
coalgebra $(\Alg,\Delta)$ can be any unital algebra. However, following
\cite{BallesterosRagnisco1998} we are interested to the case when the
algebra $\Alg$ is a \emph{Poisson algebra}.  We have then the following
definition:

\begin{definition}
    The pair $(\mathfrak{A}, \Delta)$ is a \emph{Poisson coalgebra} if
    $\Alg$ is a Poisson algebra and $\Delta$ is a Poisson homomorphism
    between $\Alg$ and $\Alg \otimes \Alg$, i.e.:
    \begin{equation}
    \Delta\bigl(\{X, Y\}_\mathfrak{A}\bigl)=\{\Delta(X),\Delta(Y)\}_{\mathfrak{A} \otimes\mathfrak{A}} \quad \forall \, X,Y \in \mathfrak{A} \, ,
    \label{eq:hompoi}
    \end{equation}
    with respect to the standard Poisson structure on $\Alg \otimes \Alg$ 
    given by:
    \begin{equation}
        \{X \otimes Y, W \otimes Z\}_{\mathfrak{A} \otimes \mathfrak{A}}:=\{X, W\}_{\mathfrak{A}} \otimes YZ+XW \otimes \{Y,Z\}_{\mathfrak{A}},
        \label{eq:pstrucatensa}
    \end{equation}
    for all $X,Y,Z,W\in\Alg$.
    \label{key}
\end{definition}
So, let $(\Alg,\Delta)$ be a Poisson coalgebra generated by the set 
$\left\{ A_{1},\dots,A_{K} \right\}$, with $K:=\text{dim}(\Alg)$.
In what follows we will denote by $C_{i} = C_{i}\left( A_{1},\dots,A_{K}
\right)$ for $i=1,\dots,r$ the $r$ functionally independent \emph{Casimir functions} of $\Alg$.

\begin{remark}
    We remark that also when $\Alg$ is a Poisson algebra, the definition
    of the coproduct map might be non-trivial.  However, for Lie--Poisson
    algebras with generators $A_i$, $i = 1,\dots,K$, the coproduct is 
    primitive \cite{Tjin1992}, i.e.:
    \begin{equation}
        \Delta(A_i)= A_i\otimes 1+1\otimes A_i \qquad \Delta(1)=1\otimes 1,
        \label{eq:standardcop}
    \end{equation}
    and for general elements it is defined by extension.
    %We call such algebras \emph{Lie--Poisson algebras}.
    \label{rem:liepoisson}
\end{remark}

%From direct computation we have that the {coassociative coproduct} 
%$\Delta$ is  a Poisson map with respect to 
%the standard Poisson bracket on $A\otimes A$.
%Indeed, using the definitions and the coassociativity property we
%have:
%\begin{equation}
%    \{A_i\otimes A_j,A_r\otimes A_s\}_{\Alg\otimes  \Alg}=
%    \{A_i, A_r\}_\Alg \otimes A_j  A_s +
%    A_i  A_r \otimes \{A_j, A_s \}_\Alg,
%    \label{eq:delta2}
%\end{equation}
%which proves the statement.

By recursion we can define the $m$th coproduct 
$\Delta^{(m)}\colon \Alg \rightarrow \Alg^{\otimes m}$, 
where $\Alg^{\otimes m}=\bigotimes_{i=1}^{m}\Alg$, as:
\begin{equation}
    \Delta^{(m)}:=
    \begin{cases}
        \Delta & \text{if $m=2$},
        \\
        \bigl(\overbrace{\Id\otimes \Id\otimes \dots \otimes \Id}^{m-2}\otimes \Delta^{(2)}\bigr)\circ\Delta^{(m-1)}.
        &
        \text{if $m>2$}.
    \end{cases}
    \label{eq:deltam}
\end{equation}
By induction on $m$, from the fact that $\Delta$ is a Poisson map
on $\Alg\otimes\Alg$ we have that the $m$th coproduct $\Delta^{(m)}$
is a Poisson map on $\Alg^{\otimes m}$ \cite{BallesterosRagnisco1998,
Ballesteros_et_al2009} (see also \cite{Musso2010loop}).

Using the $m$th coproduct it is possible to define elements on
$\Alg^{\otimes m}$ by applying it to the generators $A_{i}$
of the original Poisson algebra $\Alg$.
This can be extended to functions $h\in \mathcal{C}^{\infty}(\Alg)$ 
through the following formula:
\begin{equation} 
    h^{(m)}:=\Delta^{(m)}(h)(A_1,\dots,A_K):=
    h(\Delta^{(m)}(A_1),\dots,\Delta^{(m)}(A_K)).
    \label{eq:htotg}
\end{equation}
Clearly $h^{(m)}\in \mathcal{C}^{\infty}(\Alg^{\otimes m})$.

\subsection{Coalgebra symmetry and Liouville integrability} 
The extension we just proposed can be performed to the Casimir elements
to produce a set of commuting invariants on a given realisation in $N$
degrees of freedom.  This is the first link of the coalgebra structure
with Liouville integrability.

To be more specific, let us take a $N\in\N$ fixed.  Then we can consider
a chain of tensor products $\Alg^{\otimes m}$ for all $m \leq N$ with
embedding:
\begin{equation}
    j_{m} 
    \colon 
    A\in \Alg^{\otimes m} \mapsto A \otimes \overbrace{1\otimes 1\otimes \dots \otimes 1}^{N-m}
    \in \Alg^{\otimes N},
    \label{eq:Algembed}
\end{equation}
so that we can consider all the elements as lying in the final tensor
product $\Alg^{\otimes N}$.  In particular, see \cite{Musso2010gaudin},
we can consider all the Poisson brackets in $\Alg^{\otimes N}$ in the
following way: let $A\in\Alg^{\otimes m}$ and $A'\in\Alg^{\otimes m'}$,
then:
\begin{equation}
    \left\{ A, A' \right\}_{\Alg^{\otimes N}} = 
    \left\{ j_{m}(A),j_{m'}(A') \right\}_{\Alg^{\otimes N}}.
    \label{eq:algNpoisson}
\end{equation}

We state the following:

\begin{lemma}[\cite{BallesterosRagnisco1998}]
    Consider the following $r N$ functions:
    \begin{equation}
        \mathcal{F}_{m,j}:= \Delta^{(m)}(C_j)(A_1,\dots,A_K),
        \quad m=1,\dots,N,j=1,\dots,r.
        \label{eq:Ctotg}
    \end{equation}
    Then the set $\mathcal{L}_{\mathsf{r},N} = \left\{ \mathcal{F}_{m,j}
    \right\}_{m=1,\dots,N,j=1,\dots,r}$ is a set of Poisson-commuting
    functions on $\Alg^{\otimes N}$. Furthermore, they Poisson commute
    with $\Delta^{(N)}(A_i)$, $i=1, \dots, K$.  
    \label{lem:comm}
\end{lemma}

\begin{proof}
    We refer the reader to \cite{BallesterosRagnisco1998, Ballesteros_et_al2009}. 
\end{proof}

In the continuum case, one can use formula \eqref{eq:htotg} to
define a $N$-degrees of freedom Hamiltonian starting from a one
degree of freedom Hamiltonian.  That is, taking a function $h\in
\mathcal{C}^{\infty}(\Alg)$, we generate its $N$ degrees of freedom
extension by defining $h^{(N)}\in \mathcal{C}^{\infty}(\Alg^{\otimes
N})$ from equation \eqref{eq:htotg}.  A Hamiltonian constructed in
this way is said to admit the \emph{coalgebra symmetry}, and by taking
into account Lemma \ref{lem:comm}, it is possible to show that Poisson
commutes with the functions $\mathcal{F}_{m,j}$ making them invariants
(first integrals) of the associated dynamical system of Hamilton equations
\cite{BallesterosRagnisco1998, Ballesteros_et_al2009}.

As underlined in the previous section, in the discrete setting there
is no exact discrete analog of the Hamiltonian, so the construction
we just explained does not extend trivially. However, we consider the
evolution of the generator of the Poisson algebra $\Alg$ under the
flow of a discrete dynamical system. In such a case to underline the
discrete nature of the evolution we specify the Poisson algebra as
$\Alg = \langle A_{1}^{(n)},\dots,A_{K}^{(n)}\rangle$, and work with
a given \emph{symplectic realisation} of $\Alg$. We then state the
following definition:

\begin{definition}
    A Poisson map $T\colon n\mapsto n+1$ is said to possess the
    \emph{coalgebra symmetry} with respect to the Poisson coalgebra
    $(\Alg,\Delta)$ if for all $N\in\N$ the evolution of generators in
    a fixed symplectic realisation in $N$ degrees of freedom of the 
    Poisson coalgebra is:
    \begin{enumerate}[label=(\roman*)]
        \item \emph{closed} in the Poisson coalgebra, that is: 
            \begin{equation}
                A_{i}^{(n+1)}  = a_{i}(A_{1}^{(n)},\dots,A_{K}^{(n)}), 
                \quad i=1,\dots,K,
                \label{eq:closurerelation}
            \end{equation} 
            with $a_{i}\in \mathcal{C}^{\infty}(\Alg)$,
        \item it is a Poisson map with respect to the Poisson algebra $\Alg$,
            that is:
            \begin{equation}
                \{ A_{i}^{(n+1)} ,A_{j}^{(n+1)}\} = T\bigl(\{A_{i}^{(n)},A_{j}^{(n)} \}\bigr),
                \quad i,j=1,\dots,K,
                \label{eq:poissonalg}
            \end{equation}
        \item admits the Casimirs $\{ C_{1}^{(n)},\dots,C_{r}^{(n)} \}$
            of the algebra $\Alg$ as invariants.
    \end{enumerate}
    \label{def:dcoalgebrasymmety}
\end{definition}

So, now note immediately that any invariant for the system
\eqref{eq:closurerelation} is an invariant for the original Poisson map.
This trivially follows noticing that the system \eqref{eq:closurerelation}
on $N$ degrees of freedom symplectic realisation of $\Alg$ is a
consequence of the Poisson map itself: using the coordinate realisation
we are able to build the invariants in the original coordinates.
This construction, enables us to summarise our construction of the
discrete coalgebra symmetry in the following theorem:

\begin{theorem}
    Consider a Poisson map $T\colon n\mapsto n+1$ with coalgebra symmetry with 
    respect to the coalgebra $(\Alg,\Delta)$.
    Then $T\colon n\mapsto n+1$ admits a set of $rN$ 
    commuting invariants $\mathcal{L}_{r,N}$ 
    \eqref{eq:Ctotg}.
    \label{thm:coalgebraconstruction}
\end{theorem}

\begin{proof}
    The proof follows trivially applying the $m$th coproduct to the system 
    \eqref{eq:closurerelation}, and then considering the construction of the functions
    $\mathcal{F}_{j,m}$ \eqref{eq:Ctotg}.
    Commutation follows from Lemma \ref{lem:comm}.
\end{proof}

\subsection{Coalgebra symmetry and superintegrability}

In what we discussed in the previous sections we applied
recursively the map $\Delta^{(2)}$ to define the $m$-th coproduct
maps $\Delta^{(m)}$. However, besides \eqref{eq:deltam} another
recursive relation can be defined for the $m$-th coproduct maps,
i.e. \cite{Ballesteros_et_al2009}:
\begin{equation}
    \Delta^{(m)}_{R}:=
    \begin{cases}
        \Delta & \text{if $m=2$},
        \\
        \bigl(\Delta^{(2)}\otimes\overbrace{\Id\otimes \Id\otimes \dots \otimes \Id}^{m-2}\bigr)\circ\Delta^{(m-1)}.
        &
        \text{if $2<m\leq N$}.
    \end{cases}
    \label{eq:rightCas}
\end{equation}
Since the coproduct is coassociative, we have that for all $N\in\N$ 
fixed:
\begin{equation}
    \Delta^{(N)}(A_i)=\Delta_R^{(N)}(A_i).
    \label{eq:coas}
\end{equation}

However, when lower dimensional coproducts are considered, substantial
differences arise. In particular, lower dimensional left $m$-th coproducts
with $2 < m < N$  will contain objects living on the tensor product
space $1 \otimes 2 \otimes \dots \otimes m$, whereas lower dimensional
right $m$-th coproducts will be defined on the sites $(N  -m+1) \otimes
(N -m+2)\otimes \dots \otimes N$.  This implies that coalgebra symmetry
not only generate a completely integrable Hamiltonian system but, in principle, a superintegrable one. This is because besides the left
Casimirs \eqref{eq:Ctotg} we will also be able to define the set 
$\mathcal{R}_{r,N} = \{\mathcal{G}_{m,j} \}_{m=1,\dots,N,j=1,\dots,r}$
composed by the functions
\begin{equation}
    \mathcal{G}_{m,j}:= \Delta^{(m)}_{R}(C_j)(A_1,\dots,A_K),
    \quad m=1,\dots,N,j=1,\dots,r.
    \label{eq:rfun}
\end{equation}
called the \emph{right Casimirs}.  Analogously than Lemma
\ref{lem:comm} we have that this set is composed by $r N$ functions
in involution.  Notice that because of the coassociativity property
$\mathcal{F}_{N,j}=\mathcal{G}_{N,j}$, $j=1, \dots r$.

Summing up we obtain the following result:

\begin{theorem}
    Consider a Poisson map $T\colon n\mapsto n+1$ with coalgebra symmetry
    with respect to the coalgebra $(\Alg,\Delta)$.  Then $T\colon n\mapsto
    n+1$ admits two set of $rN$ commuting invariants: $\mathcal{L}_{r,N}$
    \eqref{eq:Ctotg} and $\mathcal{R}_{r,N}$ \eqref{eq:rfun}.
    \label{thm:coalgebrasupint}
\end{theorem}

\begin{remark}
    We remark that the elements of the sets $\mathcal{L}_{r,N}$ and
    $\mathcal{R}_{r,N}$ in general do not commute between themselves,
    see for instance \cite{Latini_et_al2021} for the example of the
    $\Sl_{2}(\RR)$ Lie--Poisson algebra.
    \label{rem:comm}
\end{remark}

\subsection{Final remarks on the coalgebra symmetry}
Before discussing some concrete cases with explicit Lie--Poisson algebras
we give some final remarks on the procedure: 
\begin{enumerate}
    \item Even in the continuum case it is not possible to tell
        \emph{a priori} if a system with coalgebra symmetry is Liouville
        integrable, see \cite{BallesterosBlasco2008}.
        This depends on the explicit symplectic realisation of the coalgebra
        $(\Alg,\Delta)$ chosen and the number of functionally independent 
        invariants that is possible to extract from the set 
        $\mathcal{L}_{r,N}$.
        When we have $N$ degrees of freedom it is enough that the set 
        $\mathcal{L}_{r,N}$ consists of $N-1$ functionally independent 
        invariants: the coalgebraic Hamiltonian $h^{(N)}$ 
        will yield the last one.
    \item In all the examples we will present, even though the 
        set $\mathcal{L}_{r,N}$ consists of $N-1$ functionally independent invariants, it will not be enough to give integrability because of the lack of the Hamiltonian.
        However, in most cases the additional missing invariants
        can be found directly studying the system \eqref{eq:closurerelation}
        with the methods discussed in Appendix \ref{app:findinv}.
\end{enumerate}

\section{General classes of additive differential systems and coalgebra}
\label{sec:general}

Consider the following system of second-order additive difference equations:
\begin{equation}
    \vec{x}_{n+1}+\vec{x}_{n-1} = \vec{F}\left( \vec{x}_{n} \right).
    \label{eq:2ndadd}
\end{equation}
This system preserves the canonical $\RR^{2N}$ measure:
\begin{equation}
    m = \ud \vx_{n}\wedge\ud\vx_{n-1}.
    \label{eq:canonical}
\end{equation}
The system \eqref{eq:2ndadd} is not Lagrangian for all
choices of the vector function $\vec{F}\left( \vx_{n} \right)$.
We note that, if there exists a scalar function $V=V(\vx_{n})$
such that $\grad V(\vx_{n}) = \vec{F}(\vx_{n})$, then the
system \eqref{eq:2ndadd} can be derived by the following Lagrangian:
\begin{equation}
    L= \vec{x}_{n+1}\cdot\vec{x}_{n} - V\left( \vec{x}_{n} \right).
    \label{eq:2naddlagr}
\end{equation}
From Lemma \ref{lem:poisson} we have that to the Lagrangian \eqref{eq:2naddlagr} 
corresponds the canonical (full-rank) Poisson bracket:
\begin{equation}
    \left\{ x_{n,i},x_{n,j} \right\}=\left\{ x_{n-1,i},x_{n-1,j} \right\}=0,
    \quad
    \left\{ x_{n,i},x_{n-1,j} \right\} = \delta_{i,j}
    \label{eq:canonicalpoisson}
\end{equation}
where $\delta_{i,j}$ is Kronecker delta.

We now pass to consider two particular cases of systems in the form
\eqref{eq:2ndadd} and prove that they possess the coalgebra symmetry.

\subsection{Radially symmetric systems and the $\mathfrak{sl}_{2}(\RR)$ algebra}

Consider the following particular case of \eqref{eq:2ndadd}:
\begin{equation}
    \vec{x}_{n+1}+\vec{x}_{n-1} = f\left( \abs{\vec{x}_{n}} \right)\frac{\vec{x}_{n}}{\abs{\vec{x}_{n}}}.
    \label{eq:radsym}
\end{equation}
This system has \emph{radial symmetry}: 
if $\vec{x}_{n}$ is a solution of \eqref{eq:radsym} 
then $\vec{X}_{n}=R\vec{x}_{n}$ with $R\in\SO\left( N \right)$ 
is another solution.
The system \eqref{eq:radsym} is variational with the following
Lagrangian:
\begin{equation}
    L= \vec{x}_{n+1}\cdot\vec{x}_{n} - V\left( \abs{\vec{x}_{n}} \right),
    \quad V\left( r \right) = \int f\left( r \right)\ud r.
    \label{eq:Lrad}
\end{equation}
From the general case we have that equation \eqref{eq:radsym}
preserves the canonical Poisson bracket \eqref{eq:canonicalpoisson}.

The following result follows from a trivial computation:
\begin{lemma}
    A radially symmetric system \eqref{eq:radsym}
    possesses the following $N\left( N-1 \right)/2$ invariants:
    \begin{equation}
        L_{i,j}^{(n)} = x_{n,i}x_{n-1,j}-x_{n-1,i}x_{n,j}.
        \label{eq:Lij}
    \end{equation}
    In total there are $2N-3$ functionally independent functions 
    in formula \eqref{eq:Lij}.
    \label{lem:rad}
\end{lemma}

The skewsymmetric invariants $L_{i,j}$ are the discrete analogue of
the components of the angular momentum. It is well known that out
of the $2N-3$ functionally independent invariant it is possible to
construct $N-1$ functionally independent and commuting with respect
to the Poisson bracket \eqref{eq:canonicalpoisson}.  We show how to
construct this set of $N-1$ commuting invariants with the coalgebra
symmetry approach, as it was proved in the in the continuum case in
\cite{BallesterosHerranz2001,Ballesteros_et_al2009}. This is the content
of the following proposition:

\begin{proposition}
    A radially symmetric system \eqref{eq:radsym} possesses coalgebra
    symmetry with respect to the Lie algebra $\Sl_{2}(\RR)$ with the following
    $N$ degrees of freedom symplectic realisation:
    \begin{equation}
        J_{+}^{(n)} = \vec{x}_{n}^{2},
        \quad
        J_{-}^{(n)} = \vec{x}_{n-1}^{2},
        \quad
        J_{3}^{(n)} = \vec{x}_{n}\cdot\vec{x}_{n-1}.
        \label{eq:Nsl2}
    \end{equation}
    \label{lem:sl2general}
\end{proposition}

\begin{remark}
   The commutation relations of the Lie-Poisson algebra $\Sl_{2}(\RR)$ are
    (for sake of simplicity when dealing with abstract properties we omit the
    superscript $(n)$):
    \begin{equation} 
        \poisson{J_{+},J_{-}} = 4J_{3},
        \quad
        \poisson{J_{+},J_{3}} = 2J_{+},
        \quad
        \poisson{J_{-},J_{3}} = -2J_{-}.
        \label{eq:sl2comm}
    \end{equation}
    This algebra has a single Casimir given by 
    \begin{equation}
        C=J_{+}J_{-}-J_{3}^{2}.
        \label{eq:sl2casimir}
    \end{equation}
    \label{rem:sl2data}
\end{remark}

\begin{proof}
    We have to prove that the three conditions of Definition
    \ref{def:dcoalgebrasymmety} hold.  This can be done by direct
    computation.  For instance, using the explicit form of the recurrence
    \eqref{eq:radsym} we prove that the action on the generators of the
    $\Sl_{2}(\RR)$ algebra \eqref{eq:Nsl2} form the following dynamical
    system:
    \begin{subequations} 
        \begin{align}
            J_{+}^{(n+1)} &=
            f^{2}( \sqrt{J_{+}^{(n)}}) 
            - 2J_{3}^{(n)} \frac{f( \sqrt{J_{+}^{(n)}})}{\sqrt{J_{+}^{(n)}}}
            +
            J_{-}^{(n)},
            \label{eq:Jnplussl2gen}
            \\
            J_{-}^{(n+1)} &= J_{+}^{(n)},
            \label{eq:Jnminussl2gen}
            \\
            J_{3}^{(n+1)} &= -J_{3}^{(n)} +
            \sqrt{J_{+}^{(n)}}f( \sqrt{J_{+}^{(n)}}).
            \label{eq:Ksl2gen}
        \end{align}
        \label{eq:sl2evolutiongeneral}
    \end{subequations}
    Then using the commutation relations \eqref{eq:sl2comm} we prove
    that they are preserved (see also formula \eqref{eq:Poisson1}).
    Finally, it is trivial to see that the Casimir function 
    \eqref{eq:sl2casimir} is preserved by the evolution 
    \eqref{eq:sl2evolutiongeneral}.
\end{proof}

So, following the procedure outlined in Section
\ref{sss:coalgebra}, and using the same reasoning in
\cite{BallesterosHerranz2001,Ballesteros_et_al2009}, from the left and
right Casimir functions of such an algebra we derive the following two
sets of functionally-independent invariants commuting with respect to
the canonical Poisson bracket \eqref{eq:canonicalpoisson}:
\begin{subequations}
    \begin{align}
        \mathcal{C}_{m}^{(n)} &=
        \sum_{1\leq i<j\leq m} \left(L_{i,j}^{(n)}\right)^{2}, 
        \quad
        m=2,\dots,N,
        \label{eq:Cm}
        \\
        \mathcal{D}_{m}^{(n)} &=
        \sum_{N-m+1\leq i<j\leq N} \left(L_{i,j}^{(n)}\right)^{2}, 
        \quad 
        m=2,\dots,N.
        \label{eq:Dm}
    \end{align}
    \label{eq:casimirsl2}%
\end{subequations}
Notice that $\mathcal{C}_{N}^{(n)}=J_{+}^{(n)}J_{-}^{(n)}-(J_{3}^{(n)})^{2}=\mathcal{D}_{N}^{(n)}$ 
because of the coassociativity.
This can be summarised in the following theorem:

\begin{theorem} 
    A radially symmetric system \eqref{eq:radsym} is quasi-integrable
    with one of the sets of invariants:
    \begin{equation}
        \mathcal{L} =
        \bigl\{\mathcal{C}_{2}^{(n)},\dots,\mathcal{C}_{N}^{(n)}\bigr\}
        \quad\text{or}\quad
        \mathcal{R}=
        \bigl\{\mathcal{D}_{2}^{(n)},\dots,\mathcal{D}_{N}^{(n)}\bigr\},
        \label{eq:CDsets}
    \end{equation}
    Additionally, if we can find an additional invariant commuting either 
    with $\mathcal{L}$ or $\mathcal{R}$, then system becomes Liouville integrable 
    and moreover quasi-maximally superintegrable.
    \label{thm:radialquasiint}
\end{theorem}

As we noted in Section \ref{sss:coalgebra}, the invariants of the
system \eqref{eq:sl2evolutiongeneral} are invariants of the radially
symmetric system \eqref{eq:radsym}.
So, for a given function $f$ studying the system \eqref{eq:sl2evolutiongeneral}
we can find the $N$th invariant mentioned in Theorem \ref{thm:radialquasiint}.
We now give an explicit example of this occurrence.

\begin{example}
    Consider the following linear system:
    \begin{equation}
        \vx_{n+1} + \vx_{n-1} = \alpha \vx_{n}.
        \label{eq:linear}
    \end{equation}
    This system is radial with $f\left( \rho \right) =\alpha \rho$ 
    (hence it is variational),
    %Lagrangian:
    %\begin{equation}
    %    L = \vx_{n+1}\cdot\vx_{n} - \frac{\alpha}{2} \vx_{n}^{2}.
    %    \label{eq:linearL}
    %\end{equation}
    and the associated dynamical system \eqref{eq:sl2evolutiongeneral} is:
    \begin{equation}
        \begin{gathered}
            J_{+}^{(n+1)} =
            \alpha^{2} J_{+}^{(n)}- 2\alpha J_3^{(n)}+J_{-}^{(n)},
            \\
            J_{-}^{(n+1)} = J_{+}^{(n)},
            \quad
            J_3^{(n+1)} = -J_3^{(n)} +
            \alpha J_{+}^{(n)}.
        \end{gathered}
        \label{eq:sl2evolutionli}
    \end{equation}
    This system is linear and we find the invariant:
    \begin{equation}
        \mathcal{H}^{(n)} = J_{+}^{(n)} - \alpha J_3^{(n)} + J_{-}^{(n)}.
        \label{eq:H1lin}
    \end{equation}
    From Theorem \ref{thm:radialquasiint} we consider the set of invariants:
    \begin{equation}
        \mathcal{S} = 
        \left\{\mathcal{H}^{(n)}, \mathcal{C}_{2}^{(n)},\dots,\mathcal{C}_{N}^{(n)}\right\}.
        \label{eq:invariantslin}
    \end{equation}
    Functional independence and involutivity in this set can be proved
    by induction.
    So, we proved that the linear system \eqref{eq:linear} is
    Liouville integrable.
    Moreover, since there are $2N-3$ functionally independent elements of the 
    discrete angular momentum that are invariants, we have that the system
    \eqref{eq:linear} is quasi-maximally superintegrable.
    \hfill $\square$
\end{example}

\begin{remark}
    We remark that, with respect to the rank 2 Lie--Poisson bracket
    \eqref{eq:sl2comm} the system \eqref{eq:sl2evolutionli} is
    Poisson--Liouville integrable.  Indeed, it possesses the Casimir
    \eqref{eq:sl2casimir}, and one invariant \eqref{eq:H1lin}.
    \label{rem:intsl2lin}
\end{remark}

\begin{remark}
    The linear system \eqref{eq:linear} is actually maximally
    superintegrable.  
    Indeed it possesses the following invariants:
    \begin{equation}
        \mathcal{H}^{(n)}_{k} = x_{n,k}^{2} -\alpha x_{n,k} x_{n-1,k} + x_{n-1,k}^{2},
        \quad k =1,\dots,N.
        %J_{+}^{(n)} - \alpha J_3^{(n)} + J_{-}^{(n)}.
        \label{eq:Hk}
    \end{equation}
    The set 
    \begin{equation}
        \mathcal{S}' = 
        \left\{\mathcal{H}^{(n)}_{1},\dots, \mathcal{H}^{(n)}_{N},
        \mathcal{C}_{2}^{(n)},\dots,\mathcal{C}_{N}^{(n)}\right\}.
        \label{eq:invariantslinms}
    \end{equation}
    is a set of $2N-1$ functionally independent and commuting invariants.
    This is easily seen by induction because the invariants are polynomial
    and they all depend on different $x_{n,k}$.
    Indeed, the system is a discrete analog of an isotropic harmonic 
    oscillator, which is a well known maximally superintegrable system.
    %In fact, also the invariants \eqref{eq:Hk} have coalgebraic origin,
    %and they originate through the same one degree of freedom realisation
    %$D\colon \mathfrak{sl}_{2}(\RR) \to \RR^{2}$ given by
    %\begin{equation}
    %    D(J_{+}) = x_{n,1}^{2}, \quad D(J_{-}) = x_{n-1,1},
    %    \quad
    %    D(J_{3}) = x_{n,1}x_{n-1,1}.
    %    \label{eq:D1dof}
    %\end{equation}
    %and the corresponding one degrees of freedom system has the invariant
    %\begin{equation}
    %    \mathcal{H}_{1}^{(n)}(x_{n,1},x_{n-1,1}) = 
    %    J_{+}^{(n)}(x_{n,1},x_{n-1,1}) - \alpha J_3^{(n)}(x_{n,1},x_{n-1,1}) 
    %    + J_{-}^{(n)}(x_{n,1},x_{n-1,1}),
    %    \label{eq:H1origin}
    %\end{equation}
    %that is the same functional form as \eqref{eq:H1lin}.
    %The invariants come from the following embedding of this invariant
    %into the $N$ degrees of freedom realisation:
    %\begin{equation}
    %    \mathcal{H}^{(n)}_{k} = 
    %    \underbrace{1\otimes \dots \otimes 1 }_{k-1} 
    %    \mathcal{H}_{1}^{(n)}(x_{n,k},x_{n-1,k})
    %    \underbrace{1\otimes \dots \otimes 1 }_{N-k}.
    %    \label{eq:Hkorig}
    %\end{equation}
    \label{rem:ms}
\end{remark}

\subsection{Quasi-radially symmetric systems and the $h_{6}$ algebra}
Consider the following system that is a particular case of \eqref{eq:2ndadd}:
\begin{equation}
    \vec{x}_{n+1}+\vec{x}_{n-1} = f\left( \abs{\vec{x}_{n}} \right)\frac{\vec{x}_{n}}{\abs{\vec{x}_{n}}}
    +g\left( \vec{x}_{n} \right)\vbeta,
    \label{eq:qradsym}
\end{equation}
where $g$ is a scalar function and $\vbeta$ is a constant vector.
Since when $g\equiv0$ equation \eqref{eq:qradsym} reduces to
\eqref{eq:radsym} we call such a system of difference equations a
\emph{quasi-radially symmetric system}.

The following result follows from a trivial computation:
\begin{lemma}
    A quasi-radially symmetric system of the form \eqref{eq:qradsym}
    possess the following  $N\left( N-1 \right)\left( N-2 \right)/6$ invariants:
    \begin{equation}
        K^{(n)}_{i,j,k} = \beta_{i} L_{j,k}^{(n)}
        +\beta_{j} L_{k,i}^{(n)}+\beta_{k} L_{i,j}^{(n)},
        \quad
        1\leq i<j<k\leq N.
        \label{eq:Kijk}
    \end{equation}
    In total there are $2N-5$ functionally independent functions 
    in formula \eqref{eq:Kijk}.
    \label{lem:qrad}
\end{lemma}

Quasi-radial systems are not always variational.  We give the following
characterisation (whose proof follows from direct computations):

\begin{lemma}
    A quasi-radially symmetric system \eqref{eq:qradsym} preserves 
    the canonical Poisson bracket \eqref{eq:canonicalpoisson} if and 
    only if the
    $g\left( \vec{x}_{n} \right)\equiv h\left( \vbeta\cdot\vec{x}_{n} \right)$,
    that is:
    \begin{equation}
        \vec{x}_{n+1}+\vec{x}_{n-1} = 
        f\left( \abs{\vec{x}_{n}} \right)\frac{\vec{x}_{n}}{\abs{\vec{x}_{n}}}
        +h\left( \vbeta\cdot \vec{x}_{n} \right)\vbeta,
        \label{eq:qradsymlagr}
    \end{equation}
    In such case, the system is variational with Lagrangian:
    \begin{equation}
        L = \vec{x}_{n+1}\cdot\vec{x}_{n} - V\left( \abs{\vec{x}_{n}} \right)
        -H\left(\vbeta\cdot\vec{x}_{n}\right),
        \quad
        H\left( \rho \right) =\int h\left( \rho \right)\ud \rho,
        \label{eq:Lqrad}
    \end{equation}
    and $V$ is defined as in equation \eqref{eq:Lrad}.
    \label{lem:canonicalpreservation}
\end{lemma}

We show how to construct a set of $N-2$ commuting invariants with the
coalgebra approach. We will use the functions $K_{i,j,k}$ as building
blocks of the invariants.  This is the content of the following
proposition:

\begin{proposition}
    A variational quasi-radially symmetric system \eqref{eq:qradsymlagr} 
    possesses coalgebra symmetry with respect to the \emph{two-photon} algebra 
    $h_{6}$ with the following $N$ degrees of freedom symplectic realisation:
    \begin{equation}
        \begin{gathered}
            A_{+}^{(n)} = \vbeta\cdot\vec{x}_{n},
            \quad
            A_{-}^{(n)} = \vbeta\cdot\vec{x}_{n-1},
            \quad
            M^{(n)} = \vbeta^{2},
            \\
            B_{+}^{(n)} = \vec{x}_{n}^{2},
            \quad
            B_{-}^{(n)} = \vec{x}_{n-1}^{2},
            \quad
            K^{(n)} = \vec{x}_{n}\cdot\vec{x}_{n-1} -\frac{\vbeta^{2}}{2}.
        \end{gathered}
        \label{eq:Ntwophoton}
    \end{equation}
    \label{lem:twophotongeneral}
\end{proposition}

\begin{remark}
    In the $h_{6}$ Lie-Poisson algebra the element $M$ is central, while
    the other have the following commutation table \cite{BallesterosHerranz2001,Ballesteros_et_al2009,Zhang_et_al1990}:
    \begin{equation}
        \begin{array}{ccccccc}
            \{\phantom{f}, \phantom{g} \} & A_{+} & A_{-} & B_{+} & B_{-} & K 
            \\
            A_{+} & 0 & -M & 0 & -2A_{-} & -A_{+}
            \\
            A_{-} & M & 0 & 2 A_{+} & 0 &  A_{-}
            \\
            B_{+} & 0 & -2 A_{+} & 0 & -4K-2 M & -2B_{+}
            \\
            B_{-} & 2 A_{-} & 2 A_{-} & 4K+2M & 0 & 2 B_{-}   
            \\
            K & A_{+} & -A_{-} & 2 B_{+} & -2B_{-} & 0 
        \end{array}
        \label{eq:h6comm}
    \end{equation}
    This Lie-Poisson algebra has two Casimirs, one is the central element
    $M$, while the other is the quartic function:
    \begin{equation}
        C_{0} =
        \bigl[M B_{+} - A_{+}^{2}\bigr]
        \bigl[M B_{-} - A_{-}^{2}\bigr]-
        \bigl[ M K - A_{+}A_{-} + M^2/2 \bigr]^2.
        \label{eq:h6casimir}
    \end{equation}
    However, since in the expression \eqref{eq:h6casimir} $M$ can be
    factorised, we can lower the degree of the Casimir by one, and
    consider the cubic function $C=C_{0}/M$, i.e.:
    \begin{equation}
        C =
         MB_+B_- - B_+A^2_- - B_-A^2_+ - M(K + M/2)^2 + 2A_-A_+(K + M/2).
        \label{eq:h6casimirthird}
    \end{equation}
    \label{lem:twophotondata}
\end{remark}
\begin{proof}
    We have to prove that the three conditions of Definition
    \ref{def:dcoalgebrasymmety} hold.  This can be done by direct
    computation.  For instance, using the explicit form of the recurrence
    \eqref{eq:qradsymlagr} we prove that the action on the generators of the
    $h_{6}$ algebra \eqref{eq:Nsl2} form the following dynamical
    system:
    \begin{subequations} 
        \begin{align}
            A_{+}^{(n+1)} &=
            A_{+}^{(n)}\frac{f( \sqrt{B_{+}^{(n)}} )}{\sqrt{B_{+}^{(n)}}}
            + M^{(n)} h( A_{+}^{(n)} ) - A_{-}^{(n)},
            \label{eq:Anplush6gen}
            \\
            A_{-}^{(n+1)} &= A_{+}^{(n)},
            \label{eq:Anminush6gen}
            \\
            B_{+}^{(n+1)} &
            \begin{aligned}[t]
            &=
            f^{2}( \sqrt{B_{+}^{(n)}}) 
            +2\left[A_{+}^{(n)} h( A_{+}^{(n)} ) -K^{(n)}-\frac{M^{(n)}}{2}\right]
            \frac{f( \sqrt{B_{+}^{(n)}})}{\sqrt{B_{+}^{(n)}}}
            \\
            &+\left[M^{(n)}h( A_{+}^{(n)} )-2 A_{-}^{(n)}\right]h ( A_{+}^{(n)} )
            +
            B_{-}^{(n)},
            \end{aligned}
            \label{eq:Bnplush6gen}
            \\
            B_{-}^{(n+1)} &= B_{+}^{(n)},
            \label{eq:Bnminush6gen}
            \\
            K^{(n+1)} &= -K^{(n)}-M^{(n)} + A_{+}^{(n)} h( A_{+}^{(n)} )
            +\sqrt{B_{+}^{(n)}}f( \sqrt{B_{+}^{(n)}} ).
            \label{eq:Ksh6gen}
            \\
            M^{(n+1)} &=
            M^{(n)},
            \label{eq:Mh6gen}
        \end{align}
        \label{eq:h6evolutiongeneral}%
    \end{subequations}
    Then, using the commutation relations \eqref{eq:h6comm}, we prove
    that they are preserved (see also formula \eqref{eq:Poisson1}).
    Finally, it is trivial to see that the central element $M^{(n)}$ and
    the Casimir function \eqref{eq:h6casimir} is preserved by the evolution 
    \eqref{eq:h6evolutiongeneral}.
\end{proof}

So, following the procedure of Section \ref{sss:coalgebra}, and using
the same procedure as in \cite{BallesterosHerranz2001,Ballesteros_et_al2009}
we derive the following two sets of functionally independent invariants
commuting with respect to the canonical Poisson bracket \eqref{eq:canonicalpoisson}:
\begin{subequations}
    \begin{align}
        \mathcal{I}_{m}^{(n)} &=
        \sum_{1\leq i<j<k\leq m} \left(K_{i,j,k}^{(n)}\right)^{2}, 
        \quad
        m=3,\dots,N,
        \label{eq:Im}
        \\
        \mathcal{J}_{m}^{(n)} &=
        \sum_{N-m+1\leq i<j<k\leq N} \left(K_{i,j,k}^{(n)}\right)^{2}, 
        \quad 
        m=3,\dots,N.
        \label{eq:Jm}
    \end{align}
    \label{eq:casimir}%
\end{subequations}
Notice that $\mathcal{I}_{N}^{(n)}=\mathcal{J}_{N}^{(n)}$ because of
the coassociativity. This can be summarised in the following theorem:

\begin{theorem} 
    A variational quasi-radially symmetric system \eqref{eq:qradsymlagr} 
    is a PLN map of order $N-2$ with one of the sets of invariants:
    \begin{equation}
        \mathcal{L} =
        \bigl\{\mathcal{I}_{3}^{(n)},\dots,\mathcal{I}_{N}^{(n)}\bigr\}
        \quad\text{or}\quad
        \mathcal{R}=
        \bigl\{\mathcal{J}_{3}^{(n)},\dots,\mathcal{J}_{N}^{(n)}\bigr\},
        \label{eq:IJsets}
    \end{equation}
    Additionally:
    \begin{itemize}
        \item if we can find \emph{one} additional invariant commuting either 
            with $\mathcal{L}$ or $\mathcal{R}$, then the system becomes 
            quasi-integrable;
        \item if we can find \emph{two} additional invariant commuting either 
            with $\mathcal{L}$ or $\mathcal{R}$, then system becomes 
            Liouville integrable and moreover superintegrable with
            $2N-3$ invariants.
    \end{itemize}
    \label{thm:qradialint}
\end{theorem}

As we noted in Section \ref{sss:coalgebra}, the invariants of the
system \eqref{eq:sl2evolutiongeneral} are invariants of the variational 
quasi-radially symmetric system \eqref{eq:qradsymlagr}.
So, for given functions $f$ and $h$, studying the system \eqref{eq:sl2evolutiongeneral}
we can search for the $(N-1)$th and the $N$th invariant mentioned in 
Theorem \eqref{thm:qradialint}.
We now give an explicit example of both cases.

\begin{example}
    Consider the following linear system:
    \begin{equation}
        \vx_{n+1} +\alpha_{0} \vx_{n} + \vx_{n-1} = 
        \left(1+ \alpha_{1}\vbeta\cdot\vx_{n}\right)\vbeta.
        \label{eq:linearqr}
    \end{equation}
    This system is clearly variational quasi-radial with 
    $f\left( \rho \right) = -\alpha_{0} \rho$,
    and $h\left( \sigma \right) = 1+ \alpha_{1}\sigma$.

    The associated dynamical system is:
    \begin{subequations}
        \begin{align}
            A_{+}^{(n+1)} &= (\alpha_1 M^{(n)} -\alpha_0) A_{+}^{(n)}+M^{(n)}-A_{-}^{(n)}, 
            \\
            A_{-}^{(n+1)} &= A_{+}^{(n)}, 
            \\
            B_{+}^{(n+1)} &
            \begin{aligned}[t]
                &=
                \alpha_0^2 B_{+}^{(n)}
                +\alpha_0\left[2 K^{(n)}+M^{(n)}-2 \left(A_{+}^{(n)}\right)^2 \alpha_1-2 A_{+}^{(n)}\right]
                \\
                &+\alpha_1^2\left(A_{+}^{(n)}\right)^2 M^{(n)} 
                +2 \alpha_1 \left(M^{(n)}- A_{-}^{(n)}\right) A_{+}^{(n+1)}
                \\
                &+M^{(n)}-2 A_{-}^{(n)}+B_{+}^{(n)},
            \end{aligned}
            \\
            B_{-}^{(n+1)} &= B_{+}^{(n)},
            \\
            K^{(n+1)} &= \alpha_1 ( A_{+}^{(n)} )^{2}
            -\alpha_0 B_{+}^{(n)} +A_{+}^{(n)}- K^{(n)}- M^{(n)},
            \\
            M^{(n+1)} &= M^{(n)}.
        \end{align}
        \label{eq:qrevolutionli}%
    \end{subequations}
    Besides the central element $M^{(n)}$ and the Casimir $C^{(n)}$,
    arising from \eqref{eq:h6casimirthird} through the realisation
    \eqref{eq:Ntwophoton}, the system \eqref{eq:qrevolutionli} has two
    additional functionally independent invariants:
    \begin{subequations}
        \begin{align}
            \mathcal{H}_{1}^{(n)} &=
            \alpha_1 A_{-}^{(n)} A_{+}^{(n)} 
            -\alpha_0 K^{(n)} 
            +A_{-}^{(n)}+A_{+}^{(n)}-B_{-}^{(n)}-B_{+}^{(n)}, 
            \label{eq:H1linqr}
            \\
            \mathcal{H}_{2}^{(n)} &
            \begin{aligned}[t]
                &= 
                \alpha_0\left( A_{-}^{(n)} A_{+}^{(n)}- K^{(n)} M^{(n)}\right)
                +\left(A_{-}^{(n)}\right)^2+\left(A_{+}^{(n)}\right)^2
                \\
                &-\left(B_{-}^{(n)} +B_{+}^{(n)}\right) M^{(n)}, 
            \end{aligned}
            \label{eq:H2linqr}
        \end{align}
        \label{eq:Hlinqr}
    \end{subequations}
    So, following Theorem \ref{thm:qradialint}, we consider the set of 
    invariants
    \begin{equation}
        \mathcal{S} = 
        \left\{\mathcal{H}_{1}^{(n)},\mathcal{H}_{2}^{(n)}, 
        \mathcal{I}_{3}^{(n)},\dots,\mathcal{I}_{N}^{(n)}\right\}.
        \label{eq:invariantslinqr}
    \end{equation}
    Functional independence and involutivity in this set can be proved
    by induction.
    So, we proved that the linear system \eqref{eq:linearqr} is 
    Liouville integrable, and moreover superintegrable with the
    $2N-3$ functionally independent invariants, considering the
    $K_{i,j,k}$ from Lemma \ref{lem:qrad}.
    \hfill $\square$
\end{example}

\begin{remark}
    We remark that, with respect to the rank 4 Lie--Poisson
    bracket \eqref{eq:h6comm} the system \eqref{eq:qrevolutionli}
    is Poisson--Liouville integrable.  Indeed, it possesses a central
    element, one Casimir \eqref{eq:h6casimirthird}, and two commuting
    invariants \eqref{eq:Hlinqr}.
    \label{rem:inth6lin}
\end{remark}

\begin{example}
    Consider the following nonlinear system:
    \begin{equation}
        \vx_{n+1} +\alpha_{0} \vx_{n} + \vx_{n-1} = 
        \left[1+ \alpha_{1}\vbeta\cdot\vx_{n} 
        +\varepsilon\left( \vbeta\cdot\vx_{n} \right)^{2} \right]\vbeta.
        \label{eq:nlinearqr}
    \end{equation}
    This system is variational quasi-radial with $f\left( \rho \right) = -\alpha_{0} \rho$ and $h\left( \sigma \right) = 1+
    \alpha_{1}\sigma+\varepsilon \sigma^{2}$.  This system is a non-integrable
    deformation of the linear system \eqref{eq:linearqr}.  We can claim
    that the system is non-integrable because it is a polynomial system
    of degree higher than one, so the sequence of degrees is either
    linear or exponential \cite{Veselov1992,BellonViallet1999}.  By a
    quick computation it is easy to see that the degrees of iterates
    of \eqref{eq:nlinearqr} is $d_{n}=2^{n}$, so that the associated
    algebraic entropy is positive.  However, we will prove using the
    coalgebra symmetry that it is quasi-integrable.

    The associated dynamical system is:
    \begin{subequations}
        \begin{align}
            A_{+}^{(n+1)} &= [\alpha_1 (M^{(n)}+\varepsilon A_{+}^{(n)}) -\alpha_0] A_{+}^{(n)}
            +M^{(n)}-A_{-}^{(n)}, 
            \\
            A_{-}^{(n+1)} &= A_{+}^{(n+1)}, 
            \\
            B_{+}^{(n+1)} &
            \begin{aligned}[t]
                &=
                \alpha_0^2 B_{+}^{(n)}
                +\alpha_0\left[2 K^{(n)}+M^{(n)}-2 \left(A_{+}^{(n)}\right)^2 \alpha_1-2 A_{+}^{(n)}\right]
                \\
                &+\alpha_1^2\left(A_{+}^{(n)}\right)^2 M^{(n)} 
                +2 \alpha_1 \left(M^{(n)}- A_{-}^{(n)}\right) A_{+}^{(n+1)}
                \\
                &+M^{(n)}-2 A_{-}^{(n)}+B_{+}^{(n)}
                +2\varepsilon(A_{+}^{(n)})^2\times
                \\
                &\phantom{(A_{+}^{(n)})^2}
                \Bigl\{(\alpha_1 M^{(n)}-\alpha_0) A_{+}^{(n)}
                    -A_{-}^{(n)}
                    +M^{(n)}\bigl[1+\varepsilon (A_{+}^{(n)})^2\bigr]
            \Bigr\}
                ,
            \end{aligned}
            \\
            B_{-}^{(n+1)} &= B_{+}^{(n)},
            \\
            K^{(n+1)} &= \alpha_1 \left( A_{+}^{(n)} \right)^{2}
            -\alpha_0 B_{+}^{(n)} +A_{+}^{(n)}- K^{(n)}- M^{(n)}+\varepsilon\left(A_{+}^{(n)}\right)^{3}.
            \\
            M^{(n+1)} &= M^{(n)},
        \end{align}
        \label{eq:qrevolutionnli}%
    \end{subequations}
    Besides the central element $M^{(n)}$ and the Casimir \eqref{eq:h6casimir} realised through \eqref{eq:Ntwophoton}, the system \eqref{eq:qrevolutionli} has one additional
    invariants, given by formula \eqref{eq:H2linqr}.
    So, following Theorem \ref{thm:qradialint}, 
    we consider the set of invariants 
    \begin{equation}
        \mathcal{S}' = 
        \left\{\mathcal{H}_{2}^{(n)},
        \mathcal{I}_{3}^{(n)},\dots,\mathcal{I}_{N}^{(n)}\right\}.
        \label{eq:invariantsnlinqr}
    \end{equation}
    Functional independence and involutivity in this set can be proved
    by induction.
    So, we proved that the linear system \eqref{eq:linearqr} is 
    quasi-integrable.
    Notice, however that despite being non-integrable the system
    possess a galore of invariants: $2N-4$ total invariants, considering
    the $K_{i,j,k}$ functions from Lemma \ref{lem:qrad}.

    We notice that from a computational point of view the (real) orbits
    of equation \eqref{eq:nlinearqr}, for $\varepsilon=0$ and $\varepsilon\neq0$,
    are very similar even for $O(1)$ values of $\varepsilon$.
    To this end, see Figure \ref{fig:ex3orbits} where we compare two cases
    with $\varepsilon=0$ and $\varepsilon=10$, but same initial conditions 
    $(\vx_{0},\vx_{-1}) = (0.1,0.1,0.1,0.1,0.1,0,0,0.1)$, and parameters 
    $\alpha_{0}=0.2$, $\alpha_{1}=0.5$, $\vbeta= (0.05,0.2,0.1,0.25)$.
    \hfill $\square$
\end{example}

\begin{figure}[ht]
    \centering
    \subfloat{\includegraphics[width=0.45\linewidth]{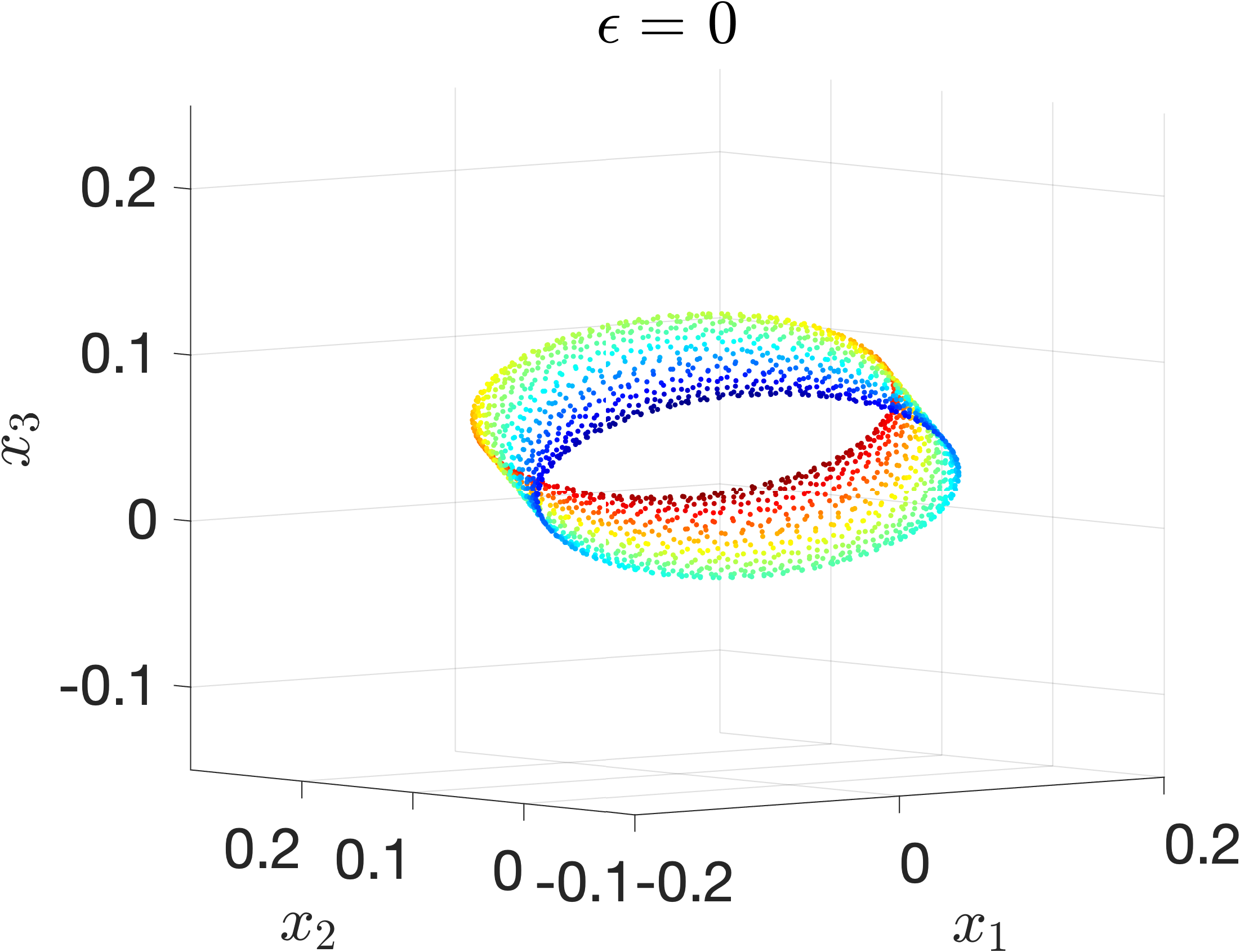}}
    \hspace{.5cm}
    \subfloat{\includegraphics[width=0.45\linewidth]{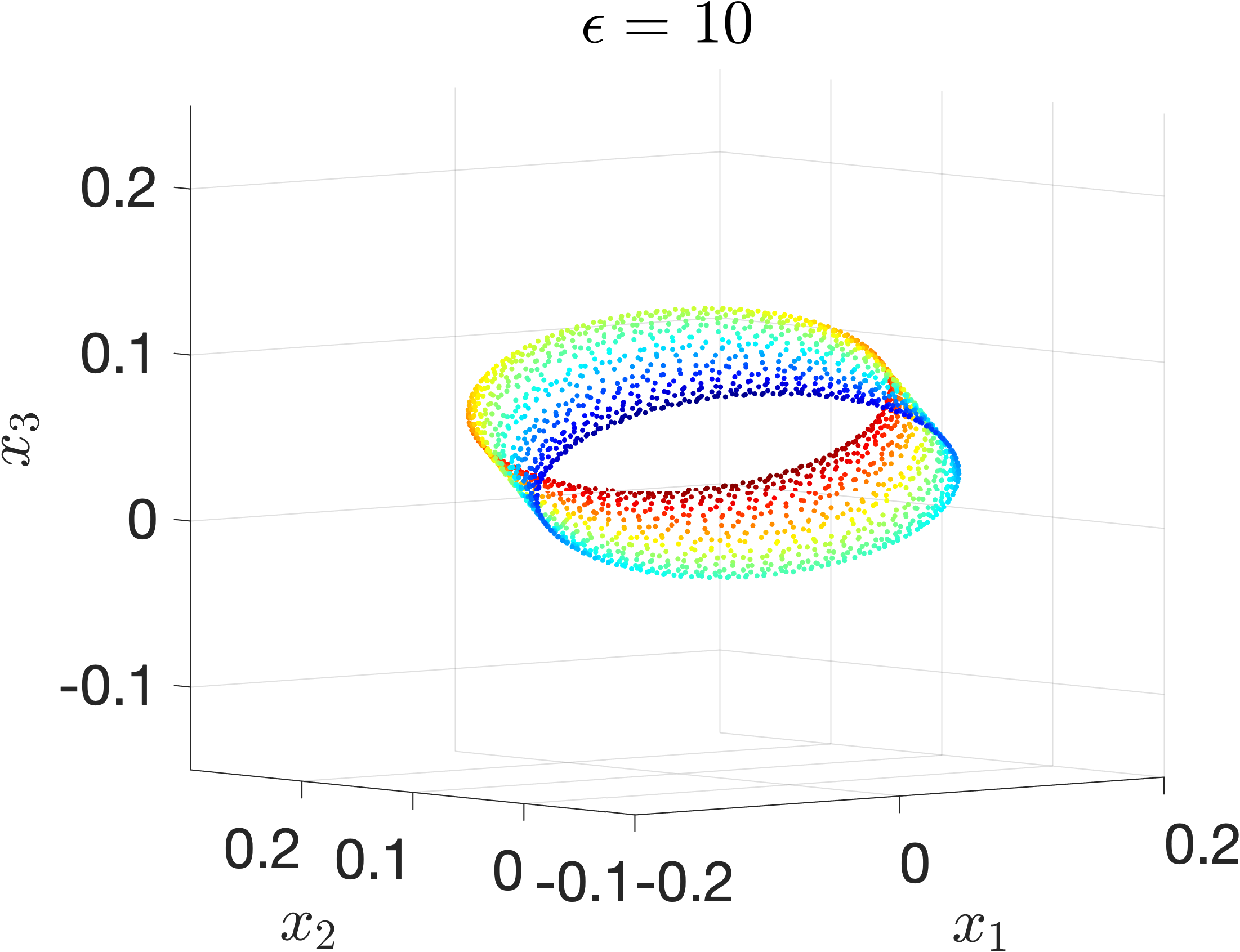}}
    \caption{Orbits of equation \eqref{eq:nlinearqr} for $\varepsilon=0$ 
    (left, equivalent to \eqref{eq:linearqr}) and equation \eqref{eq:nlinearqr} 
    for $\varepsilon=10$ (right).}
    \label{fig:ex3orbits}
\end{figure}

\begin{remark}
    We might wonder if we can hope for the existence of more 
    functionally independent invariants for the system 
    \eqref{eq:qrevolutionnli}.
    In this case algebraic entropy comes to our aid: computing the
    degrees of the iterates of the system \eqref{eq:qrevolutionnli}
    we get the following sequence of degrees:
    \begin{equation}
        1, 5, 13, 29, 61, 125, 253, 509, 1021, 2045, 4093, 8189, 16381\dots.
        \label{eq:qrevolutionnliseq}
    \end{equation}
    The growth of this sequence is clearly exponential, as it is readily
    shown by the generating function:
    \begin{equation}
        g(z) = \frac{2 z + 1}{(z - 1) (2 z - 1)},
        \label{eq:gfqrevolutionnlin}
    \end{equation}
    which implies $S= \log 2$.
    Since the algebraic entropy is less than the maximal value $\log 5$ 
    we have some form of regularity, expected by the preservation of
    the rank 4 Lie--Poisson bracket, and the existence of one central element, 
    one Casimir \eqref{eq:h6casimirthird}, and two commuting invariants \eqref{eq:Hlinqr}.
    In short, we have that the system \eqref{eq:qrevolutionnli} cannot be 
    Poisson--Liouville integrable, but it is enough regular to provide us
    one additional invariant.
    \label{rem:inth6nlin}
\end{remark}

\section{Vector extensions of the autonomous discrete Painlev\'e I}
\label{sec:dPI}

Consider the autonomous version of the discrete Painlev\'e equation
(shortly aut-$\dP_\text{I}$)
\cite{HietarintaBook,Grammaticosetal1991}:
\begin{equation}
    x_{n+1}+x_{n}+x_{n-1}=\frac{\alpha}{x_{n}}+\beta,
    \label{eq:dPI}
\end{equation}
where $\alpha,\beta\in\RR$ are arbitrary constants.
This system admits the following discrete Lagrangian \cite{Logan1973}:
\begin{equation}
    L_\text{I}^{(1)} =
    %\frac{1}{2} \left(x_{n+1}-x_{n}\right)^2-\frac{3}{2} x_{n}^2
    %+\alpha \log x_{n} +\beta x_{n},
    x_{n+1}x_{n} +\frac{x_{n}^{2}}{2}-\alpha \log x_{n} - \beta x_{n},
    \label{eq:LI1}
\end{equation}
and leaves invariant the following QRT biquadratic:
\begin{equation}
    B^{(n)} = x_{n}^{2}x_{n-1} + x_{n}x_{n-1}^{2} -\alpha \left( x_{n}+x_{n-1} \right)
    -\beta x_{n}x_{n-1}.
    \label{eq:QRTdPI}
\end{equation}

We wish to generalise this equation to $N$ degrees of freedom.
A simple observation can be made to notice that we can do this easily
by considering the Lagrangian, rather the equation itself.
Indeed, we can introduce the following two Lagrangians:
\begin{subequations}
    \begin{align}
        L_{\text{I},a}^{(N)} &= 
        %\frac{1}{2} \abs{\vx_{n+1}-\vx_{n}}^2-\frac{3}{2} \vx_{n}^2
        %+\alpha \log \abs{\vx_{n}} +\vbeta\cdot \vx_{n}.
        \vec{x}_{n+1}\cdot\vec{x}_{n} +\frac{\vec{x}_{n}^{2}}{2}
        -\alpha \log \abs{\vec{x}_{n}} -\vbeta\cdot\vec{x}_{n}.
        \label{eq:LI}
        \\
        L_{\text{I},b}^{(N)} &= 
        \vec{x}_{n+1}\cdot\vec{x}_{n} +\frac{\vec{x}_{n}^{2}}{2}
        %\frac{1}{2} \abs{\vx_{n+1}-\vx_{n}}^2
        %-\frac{3}{2} \vx_{n}^2
        -\kappa \log \valpha\cdot\vx_{n} -\vbeta\cdot \vx_{n},
    \end{align}
    \label{eq:dLIs}
\end{subequations}
where $\alpha,\kappa\in\RR$ and $\valpha,\vbeta\in\RR^{N}$.

The associated Euler--Lagrange equations are the following:
\begin{subequations}
    \begin{align}
        \vec{x}_{n+1}+\vec{x}_{n}+\vec{x}_{n-1} &=
        \frac{\alpha}{\abs{\vec{x}_{n}}^{2}}\vec{x}_{n}
        +\boldsymbol{\beta},
        \label{eq:vecdPIa}
        \\
        \vec{x}_{n+1}+\vec{x}_{n}+\vec{x}_{n-1} &=
        \kappa\frac{\valpha}{\valpha\cdot\vx_{n}}
        +\boldsymbol{\beta}.
        \label{eq:vecdPIb}
    \end{align}
    \label{eq:vecdPI}
\end{subequations}

For $2\leq N\leq4$ the heuristic calculation of the algebraic entropy for
the system \eqref{eq:vecdPIa} gives the following degree of growth:
\begin{equation}
    1, 3, 7, 15, 25, 39, 55, 75, 97, 123, 151, 183, 217, 255\dots.
    %295, 339, 385, 435, 487, 543, 601, 663, 727, 795, 865
    \label{eq:growthvecdPIa}
\end{equation}
The associated generating function is given by:
\begin{equation}
    g_{a}\left( z \right) =
    \frac{1+ z +  z^2 +3z^3}{(1+z)(1-z)^3},
    \label{eq:genfuncvecdPIa}
\end{equation}
which readily implies the growth \eqref{eq:growthvecdPIa} is quadratic.

In the same way, for $1\leq N\leq4$ the heuristic calculation of the 
algebraic entropy for the system \eqref{eq:vecdPIb} gives the following 
degree of growth:
\begin{equation}
    1, 2, 4, 8, 13, 20, 28, 38, 49, 62, 76, 92, 109\dots. 
    \label{eq:growthvecdPIb}
\end{equation}
The associated generating function is given by:
\begin{equation}
    g_{b}\left( z \right) =
    \frac{ 1+2z^3 }{(1+z)(1-z)^3},
    \label{eq:genfuncvecdPIb}
\end{equation}
which readily implies the growth \eqref{eq:growthvecdPIb} is quadratic.

These two results gave us the indication that two systems \eqref{eq:vecdPI}
are integrable.
In what follows, we will prove that these two systems are actually
superintegrable, by building their invariants.

\subsection{Superintegrability of the system \eqref{eq:vecdPIa}}
In this section, we will prove that the system \eqref{eq:vecdPIa}
is Liouville integrable for all $N\in\N$.
We start noticing that the system \eqref{eq:vecdPIa} is quasi-radially
symmetric and variational. 
So, from Theorem \ref{thm:qradialint} we have that to prove Liouville
integrability and superintegrability we need to find two additional
invariants.

Now, to build the additional invariants, we use the same
strategy we employed in the examples in Section \ref{sec:general}.
From the proof of Proposition \ref{lem:twophotongeneral} 
the associated dynamical system on the generators of the $h_{6}$ algebra
is:
\begin{subequations}
    \begin{align}
        A_{+}^{(n+1)} &=
        \alpha\frac{A_{+}^{(n+1)}}{B_{+}^{(n+1)}} + M^{(n)} -A_{-}^{(n+1)} -A_{+}^{(n+1)},
        \label{eq:ApdsdPIa}
        \\
        A_{-}^{(n+1)} &=
        A_{+}^{(n)},
        \label{eq:AmdsdPIa}
        \\
        B_{+}^{(n+1)} &
        \begin{aligned}[t]
        &=
        \frac{\left(M^{(n)}- \alpha\right) \left(2 B_{+}^{(n)}-\alpha\right) 
            +2 \alpha\left(A_{+}^{(n)}- K^{(n)}\right)}{B_{+}^{(n)}}
        \\
        &+B_{+}^{(n)}+ B_{-}^{(n)}+2 \left(K^{(n)}- A_{+}^{(n)} - A_{-}^{(n)}\right),
        \end{aligned}
        \label{eq:BpdsdPIa}
        \\
        B_{-}^{(n+1)} &=
        B_{+}^{(n)},
        \label{eq:BmdsdPIa}
        \\
        K^{(n+1)} &=
        \alpha + A_{+}^{(n)}-B_{+}^{(n)}-K^{(n)}-M^{(n)},
        \label{eq:KdsdPIa}
        \\
        M^{(n+1)} &= M^{(n)}.
        \label{eq:MdsdPIa}
    \end{align}
    \label{eq:dynsysdPIa}
\end{subequations}

Heuristically, the computation of the algebraic entropy of the dynamical 
system \eqref{eq:dynsysdPIa} gives the degree sequence \eqref{eq:growthvecdPIb}.
So, we expect to find two more commuting invariants besides the trivial central
element $M^{(n)}$ and the Casimir \eqref{eq:h6casimirthird}.
Using the method of Appendix \ref{app:findinv} we find the two commuting
invariants:
\begin{subequations}
    \begin{align}
        \mathcal{H}_{1}^{(n)} &
        \begin{aligned}[t]
            &=
            \left(\alpha+B_{+}^{(n)}-2 K^{(n)}-M^{(n)}\right) A_{-}^{(n)}
            \\
            &+\left(\alpha+B_{-}^{(n)}-2 K^{(n)}-M^{(n)}\right) A_{+}^{(n)}
            \\
            &+ K^{(n)}\left( 2 K^{(n)}+3 M^{(n)}\right)-2 B_{+}^{(n)} B_{-}^{(n)}
        \end{aligned}
        \label{eq:H1dynsysdPI}
        \\
        \mathcal{H}_{2}^{(n)} &
        \begin{aligned}[t]
            &=
            \left(K^{(n)}-A_{-}^{(n)}-A_{+}^{(n)}+\frac{B_{-}^{(n)}+B_{+}^{(n)}}{2}\right) 
            \left(\alpha^2+B_{-}^{(n)} B_{+}^{(n)}\right)
            \\
            &-2 \alpha \left( \frac{K^{(n)}}{2}+\frac{M^{(n)}}{4}\right) 
                \left[B_{-}^{(n)}+ B_{+}^{(n)} - 2\left(A_{-}^{(n)}+A_{+}^{(n)}  \right)\right]
            \\
            &-2 \alpha\left[K^{(n)}+\frac{3M^{(n)}}{2}\right]K^{(n)}
            +M^{(n)} B_{-}^{(n)} B_{+}^{(n)}.
        \end{aligned}
        \label{eq:H3dynsysdPI}
    \end{align}
    \label{eq:HdynsysdPI}%
\end{subequations}
So, following Theorem \ref{thm:qradialint}, we consider the following
set of invariants:
\begin{equation}
    \mathcal{S}_{\text{I},a} =
    \left\{ \mathcal{H}_{1}^{(n)},\mathcal{H}_{2}^{(n)},
    \mathcal{I}_{3}^{(n)},\dots,\mathcal{I}_{N}^{(n)} \right\}.
    \label{eq:Psetgen}
\end{equation}
Functional independence and involutivity in this set can be proved
by induction.
So, we proved that the system \eqref{eq:vecdPIa} is Liouville integrable
and moreover superintegrable with the $2N-3$ functionally independent 
invariants, considering the $K_{i,j,k}$ from Lemma \ref{lem:qrad}.

\begin{remark}
    We note that the same construction holds also for $N=1$,
    even though the coalgebra structure in the invariant \eqref{eq:QRTdPI}
    is not immediately clear.
    Indeed, for $N=1$ we have:
    \begin{equation}
        \mathcal{H}_{1}^{(n)} = \beta_{1}\left( B^{(n)} + M^{(n)}\right),
        \label{eq:HdynsysN1}
    \end{equation}
    proving that \eqref{eq:vecdPIa} is a \emph{bona fide} coalgebric
    extension of the autonomous discrete Painlev\'e I equation 
    \eqref{eq:dPI}.
    \label{rem:1dcasered}
\end{remark}

Now, we consider the case $\vbeta=\vec{0}$:
\begin{equation}
    \vec{x}_{n+1}+\vec{x}_{n}+\vec{x}_{n-1}=\frac{\alpha}{\abs{\vec{x}_{n}}^{2}}\vec{x}_{n}
    \label{eq:vecdPIrad}
\end{equation}
which is special because the system becomes radially symmetric.
The system is clearly integrable, but we will prove it using its
coalgebra symmetry with respect to the
$\Sl_{2}(\RR)$ algebra: from Theorem \ref{thm:radialquasiint}
if we can find an additional invariant the system becomes
quasi-maximally superintegrable.

From the proof of Proposition \ref{lem:sl2general} the associated 
dynamical system on the generators of the $\Sl_{2}(\RR)$ algebra is:
\begin{equation}
    \begin{gathered}
        J_{+}^{(n+1)} = \frac{\alpha}{J_{+}^{(n)}}\left(\alpha -2 J_{3}^{(n)} \right)
            +2 J_{3}^{(n)}-2 \alpha+J_{-}^{(n)}+J_{+}^{(n)},
        \\
        J_{-}^{(n+1)} = J_{+}^{(n)},
        \quad
        J_{3}^{(n+1)} = \alpha-\left(J_{3}^{(n)} + J_{+}^{(n)}\right).
    \end{gathered}
    \label{eq:sl2vecdPIrad}
\end{equation}

Using the method of Appendix \ref{app:findinv} we find the Casimir
\eqref{eq:sl2casimir}, and one additional invariant:
\begin{equation}
    \mathcal{H}_{1}^{(n)}\left( \vbeta=\vec{0} \right) 
    \begin{aligned}[t]
        &=
        \left(J_{+}^{(n)}+J_{-}^{(n)}+2 J_{3}^{(n)}\right) 
        \left(J_{+}^{(n)}J_{-}^{(n)}+\alpha^2\right)
        \\
        &-2 \alpha \left[2 J_{+}^{(n)}J_{-}^{(n)}+\left(J_{+}^{(n)}+J_{-}^{(n)}\right) J_{3}^{(n)}\right].
    \end{aligned}
    \label{eq:vecdPIradH1}
\end{equation}

So, following Theorem \ref{thm:radialquasiint}, we consider the following
set of invariants:
\begin{equation}
    \mathcal{S}_{\text{I},a}\left( \vbeta=\vec{0} \right) =
    \left\{ \mathcal{H}_{1}^{(n)}\left( \vbeta=\vec{0} \right),
    \mathcal{C}_{2}^{(n)},\dots,\mathcal{C}_{N}^{(n)} \right\}.
    \label{eq:Psetgenbeta0}
\end{equation}
Functional independence and involutivity in this set can be proved
by induction.
So, we proved that the system \eqref{eq:vecdPIrad} is Liouville integrable
and moreover quasi-maximally superintegrable with the $2N-2$ functionally 
independent invariants, considering the independent components of the discrete
angular momentum $L_{i,j}$ from Lemma \ref{lem:rad}.

\subsection{Superintegrability of the system \eqref{eq:vecdPIb}}

Consider the matrix $R\in \mathrm{SO}\left( N \right)$ such
that $R^{T} \boldsymbol{\hat{\alpha}} = \left( 1,0,\dots,0 \right)^{T}$,
where $\boldsymbol{\hat{\alpha}}=\valpha/\abs{\valpha}$.
So, the linear transformation $\vx_{n}= R \vec{y}_{n}$
maps the system \eqref{eq:vecdPIb} into:
\begin{subequations}
    \begin{align}
        y_{1,n+1}+y_{1,n}+y_{1,n-1} &= \frac{\kappa}{y_{1,n}}+ \beta_{1}', 
        \\
        \vec{Y}_{n+1}+\vec{Y}_{n}+\vec{Y}_{n-1} &= \vec{B},
    \end{align}
    \label{eq:vecdPIbt}
\end{subequations}
where $\vec{Y}_{n} = \left( y_{2,n},\dots,y_{N,n} \right)^{T}$,
$\boldsymbol{\beta'}=R \vbeta$, 
and $\vec{B}=\left( \beta_{2}',\dots,\beta_{N}' \right)$.
Clearly, this preserves the variational structure with the following
Lagrangian:
\begin{equation}
    \ell_{\text{I},b}^{(N)} = 
    %\frac{1}{2} \abs{\vec{y}_{n+1}-\vec{y}_{n}}^2-\frac{3}{2} \vec{y}_{n}^2
    \vec{y}_{n+1}\cdot\vec{y}_{n}+\frac{\vec{y}_{n}^{2}}{2}
    -\kappa \log y_{1,n} -\boldsymbol{\beta'} \cdot \vec{y}_{n}.
\end{equation}

The system is clearly a superposition of an autonomous discrete
Painlev\'e I equation \eqref{eq:dPI} in the variable $y_{1,n}$
and a linear quasi-radial system of the form \eqref{eq:linearqr}
with $\alpha_{0}=1$ and $\alpha_{1}=0$.
From Theorem \ref{thm:qradialint}, we derive the following set
of invariants:
\begin{equation}
    \mathcal{S}_{\text{I},b}
    =
    \left\{ B^{(n)}\left( y_{1,n},y_{1,n-1} \right), 
        \mathcal{E}_{1}^{(n)},\mathcal{E}_{2}^{(n)},
        \mathcal{K}^{(n)}_{2},\dots,
        \mathcal{K}^{(n)}_{N-1}
    \right\},
    \label{eq:setPLIbproof1}
\end{equation}
where:
\begin{equation}
    \mathcal{E}_{i}^{(n)}
    =
    \mathcal{H}_{i}^{(n)} \left( \vec{Y}_{n},\vec{Y}_{n-1},\vec{B},\alpha_{0}=1,\alpha_{1}=0 \right),
    \quad
    i=1,2,
    \label{eq:EdefdPIb}
\end{equation}
with $\mathcal{H}_{i}^{(n)}$ given in equation \eqref{eq:Hlinqr},
and
\begin{equation}
    \mathcal{K}_{i}^{(n)}
    =
    \mathcal{I}_{i}^{(n)} \left( \vec{Y}_{n},\vec{Y}_{n-1},\vec{B},\alpha_{0}=1,\alpha_{1}=0 \right),
    \quad
    i=2,\dots,N-1.
    \label{eq:KdefdPIb}
\end{equation}
with $\mathcal{I}_{i}^{(n)}$ given in equation \eqref{eq:Im}.
The functional independence and the commutation of the invariants in the set
$\mathcal{S}_{\text{I},b}$ can be proven by induction.
Moreover, from Lemma \eqref{lem:qrad} for $N-1$ we obtain $2N-5$ functionally
independent invariants.
So, in total we have (by induction) a set of $2N-4$ functionally independent
invariants, which makes the system \eqref{eq:vecdPIbt} superintegrable.

So, applying the inverse transformation $\vec{y}_{n} = R^{-1} \vx_{n}$ we
obtain that the original system \eqref{eq:vecdPIb} is superintegrable with
$2N-4$ functionally independent invariants.

\begin{remark}
    We remark that there is another candidate coalgebra symmetry for the system 
    \eqref{eq:vecdPIb}.
    That is, consider the Poisson--Lie algebra $\mathcal{A}_{10}$ generated by
    \begin{equation}
        \begin{gathered}
            A_{+}^{(n)} = \valpha\cdot\vec{x}_{n},
            \quad
            A_{-}^{(n)} = \valpha\cdot\vec{x}_{n-1},
            \quad
            B_{+}^{(n)} = \vbeta\cdot\vec{x}_{n},
            \quad
            B_{-}^{(n)} = \vbeta\cdot\vec{x}_{n-1},
            \\
            C_{+}^{(n)} = \vec{x}_{n}^{2},
            \quad
            C_{-}^{(n)} = \vec{x}_{n-1}^{2},
            \quad
            K^{(n)} = \vec{x}_{n}\cdot\vec{x}_{n-1},
            \\
            M^{(n)}_{\alpha} = \valpha^{2},
            \quad
            M^{(n)}_{\beta} = \vbeta^{2},
            \quad
            M^{(n)}_{\alpha\beta} = \valpha\cdot\vbeta,
        \end{gathered}
        \label{eq:NA}
    \end{equation}
    defined with  respect to the canonical Poisson bracket 
    \eqref{eq:canonicalpoisson}.
    Here the elements $\{M_{\alpha}^{(n)},M_{\beta}^{(n)},M^{(n)}_{\alpha\beta} \}$
    are central, while the other elements realise the following commutation table:
    \begin{equation}
        \begin {array}{cccccccc} 
            \poisson{\quad,\quad} & A_{+}^{(n)} & A_{-}^{(n)} & B_{+}^{(n)} & B_{-}^{(n)} &
            C_{+}^{(n)} & C_{-}^{(n)} & K^{(n)}
            \\
            A_{+}^{(n)} &0& M_{\alpha}^{(n)} & 0 & M_{\alpha\beta}^{(n)} & 
            0 & 2 A_{-}^{(n)} & A_{+}^{(n)}
            \\ 
            A_{-}^{(n)} & - M_{\alpha}^{(n)} & 0 & - M_{\alpha\beta}^{(n)} & 
            0 & -2 A_{+}^{(n)} &0&- A_{-}^{(n)}
            \\ 
            B_{+}^{(n)} & 0& M_{\alpha\beta}^{(n)} & 0 & M_{\beta}^{(n)} & 0 &
            2 B_{-}^{(n)}  & B_{+}^{(n)}
            \\ 
            B_{-}^{(n)} & - M_{\alpha\beta}^{(n)} & 0 & -M_{\beta}^{(n)} & 0 & 
            -2 B_{+}^{(n)} & 0 &- B_{-}^{(n)}
            \\ 
            C_{+}^{(n)} & 0 & 2 A_{+}^{(n)} & 0 & 2 B_{+}^{(n)} & 0 & 
            4 K^{(n)} & 2 C_{+}^{(n)}
            \\ 
            C_{-}^{(n)} & -2 A_{-}^{(n)} &0&-2 B_{-}^{(n)} & 0 & 
            -4 K_{n} & 0 &-2 C_{-}^{(n)}
            \\ 
            K^{(n)} & - A_{+}^{(n)} & A_{-}^{(n)} & - B_{+}^{(n)} & 
            B_{-}^{(n)} & - 2 C_{+}^{(n)} &2 C_{-}^{(n)} & 0
        \end {array}
        \label{eq:A10comm}
    \end{equation}
    So, $\mathcal{A}_{10}$ represents the Poisson analogue of a non-semisimple Lie algebra whose Levi decomposition is 
    given by:
    \begin{equation}
        \mathcal{A}_{10}
        =
        \langle A_{+}^{(n)},A_{-}^{(n)},B_{+}^{(n)},B_{-}^{(n)},
            M_{\alpha}^{(n)},M_{\beta}^{(n)},M_{\alpha\beta}^{(n)} \rangle 
        \oplus_{S}
        \langle C_{+}^{(n)}, C_{-}^{(n)},K^{(n)} \rangle
        \label{eq:levidec}
    \end{equation}
    Note that 
    $\langle C_{+}^{(n)}, C_{-}^{(n)},K^{(n)} \rangle\simeq \mathfrak{sl}_{2}(\mathbb{R})$.

    We were able to prove that the dynamical system associated with
    the generators \eqref{eq:NA} and the evolution \eqref{eq:vecdPIb}
    is closed.  Furthermore, we can prove that this associated system
    is integrable both according to the algebraic entropy criterion and
    the Liouville--Poisson definition.  However, the search for Casimir
    invariants of this Lie-Poisson algebra is still an open problem.

    Moreover, we note that the algebra can be extended to the algebra
    $\mathcal{A}_{2(2k+1)}$ for every $k\in\N$ in the following way:
    \begin{equation}
        \begin{gathered}
            A_{i,+}^{(n)} = \valpha_{i}\cdot\vec{x}_{n},
            \quad
            A_{i,-}^{(n)} = \valpha_{i}\cdot\vec{x}_{n-1},
            \quad
            K^{(n)} = \vec{x}_{n}\cdot\vec{x}_{n-1},
            \\
            Q_{+}^{(n)} = \vec{x}_{n}^{2},
            \quad
            Q_{-}^{(n)} = \vec{x}_{n-1}^{2},
            \quad
            M^{(n)}_{\alpha_{i}\alpha_{j}} = \valpha_{i}\cdot\valpha_{j},
        \end{gathered}
        \label{eq:NAk}
    \end{equation}
    where $i=1,\dots,k$ and $j=i,\dots,k$.
    A Levi decomposition similar to \eqref{eq:levidec} holds.
    The algebra $\mathcal{A}_{2(2k+1)}$, its Casimir(s), associated discrete dynamical
    systems and construction of the invariants will be subject of
    future works.
    \label{rem:addcoalg}
\end{remark}

\subsection{Continuum limits}
The autonomous discrete Painlev\'e equation \eqref{eq:dPI} has this 
name because under the following coordinate scaling:
\begin{equation}
    x_{n} = 1+\frac{ A}{3}h^2 X(t),
    \quad
    \alpha = -3 -\frac{AB}{3} h^4, 
    \quad 
    \beta = 6,
    \quad
    t = nh,
    \label{eq:clim}
\end{equation}
in the limit $h\to0$ reduces to:
\begin{equation}
    \ddot{X}= -A X^{2} - B,
    \label{eq:weierstrass}
\end{equation}
whose deautonomisation is the Painlev\'e I equation
\cite{InceBook}.
Notice that, since we are in the autonomous case,
this differential equation is related to the differential equation 
solved by the Weiestra\ss\ $\wp$-function \cite{WhittakerWatson1927}:
\begin{equation}
    \left[ \wp'\left( z \right) \right]^{2}
    =
    4\wp^{3}(z)-g_{2}\wp(z)-g_{3},
    \quad
    z,g_{2},g_{3}\in\Cp.
    \label{eq:weierstrassoriginal}
\end{equation}

Therefore, we can suppose that a similar scaling holds in the
$N$ degrees of freedom case.
However, it is possible to prove by direct computation that
no scaling of the form
\begin{equation}
    \vec{x}_{n} = \vec{x}_{0}+\frac{ A}{3}h^\gamma \vec{X}(t),
    \quad
    \alpha = \alpha(h), 
    \quad 
    \vbeta = \vbeta(h),
    \quad
    \vec{x}_{0}\in\Sph^{N},
    \quad \gamma\in\N,
    \label{eq:climN}
\end{equation}
where $\alpha(h)$ and $\vbeta(h)$ are analytic functions of their
argument, balances the terms in the systems \eqref{eq:vecdPI}.

So, at the moment, the nontrivial problem of the continuum limit of both systems \eqref{eq:vecdPI} is still open and we cannot relate those systems to vector extensions of the Weierstra\ss\ differential
equations \eqref{eq:weierstrassoriginal}.

\section{Vector extension of the McMillan map}
\label{sec:dPII}

Consider the so-called McMillan map \cite{McMillan1971}:
\begin{equation}
    x_{n+1}+x_{n-1}=\frac{\alpha x_{n}+\beta}{1-x_{n}^{2}}.
    \label{eq:mcm}
\end{equation}
This a famous discrete integrable system, admitting the following
biquadratic invariant:
\begin{equation}
    B^{(n)} =
        x_{n}^{2} x_{n-1}^{2}
        - x_{n}^{2}- x_{n-1}^{2}
        +\alpha x_{n} x_{n-1}
        +\beta \left( x_{n}+ x_{n-1} \right).
    \label{eq:mcmbiquadr}
\end{equation}
Historically, this was one of the first discrete systems ever 
introduced, and it is an example of a QRT map.
This system has the following discrete Lagrangian:
\begin{equation}
    L_\text{II} = x_{n+1}x_{n} 
    -\left[
        \frac{\alpha+\beta}{2} \log(x_{n}-1)+\frac{\alpha-\beta}{2}\log(x_{n}+1)
    \right].
    \label{eq:LII}
\end{equation}

\begin{remark}
    The McMillan equation \eqref{eq:mcm} is also known to be
    the autonomous limit of the discrete Painlev\'e II equation,
    as it was found in \cite{Grammaticosetal1991}, through the so-called
    singularity confinement method.
    \label{rem:singconf}
\end{remark}

An obvious $N$ degrees of freedom generalisation of the McMillan map
is given by the following system of equations:
\begin{equation}
    \vec{x}_{n+1}+\vec{x}_{n-1}=\frac{\alpha\vec{x}_{n}+\vbeta}{1-\vec{x}_{n}^{2}}.
    \label{eq:vecdPIIfull}
\end{equation}
Note that the discrete Lagrangian \eqref{eq:LII} does not generalise
to a Lagrangian for the system \eqref{eq:vecdPIIfull}.
Unfortunately, due to the lack of general theorems for discrete
Lagrangians of system of second order equations we cannot claim that
this system does not possess a Lagrangian at all.

Moreover, the algebraic entropy test applied to this system for 
$2\leq N \leq 4$ gives us the following degree sequence:
\begin{equation}
    1, 3, 9, 21, 45, 93, 189, 381, 765, 1533, 3069, 6141, 12285, 24573\dots,
    %49149, 98301, 196605
    \label{eq:seqmcmNcompl}
\end{equation}
which has the following generating function:
\begin{equation}
    g\left( z \right) =  \frac{(z+2)(z+1)}{(z-1)(z-2)}. 
    \label{eq:gfmcmNcompl}
\end{equation}
This suggests that the algebraic entropy of the system \eqref{eq:vecdPIIfull} 
is $S=\log 2>0$.
That is, the algebraic entropy test suggests that the system
\eqref{eq:vecdPIIfull} is not integrable.
Notice that the algebraic entropy is not maximal (because the map associated
to \eqref{eq:vecdPIIfull} has degree 3).
This immediately suggests the existence of some commuting invariants.

On the other hand, consider the case $\vbeta=\vec{0}$:
\begin{equation}
    \vec{x}_{n+1}+\vec{x}_{n-1}=\frac{\alpha\vec{x}_{n}}{1-\vec{x}_{n}^{2}}.
    \label{eq:vecdPII}
\end{equation}
We now have for $2\leq N \leq 4$ the following growth of degrees:
\begin{equation}
    1, 3, 9, 19, 33, 51, 73, 99, 129, 163, 201, 243, 289, 339, 393, 451, 513\dots.
    %579, 649, 723, 801, 883, 969, 1059, 1153\dots
    \label{eq:seqmcmNbeta0}
\end{equation}
This growth is clearly sub-exponential. Indeed, computing the generating
function we obtain:
\begin{equation}
    g\left( z;\vbeta=\vec{0} \right) =
    \frac{3z^2 + 1}{(1-z)^3},
    \label{eq:gfmcmNbeta0}
\end{equation}
which implies that the growth is quadratic.
For all $N\in\N$ the system \eqref{eq:vecdPII} is radial and possesses
the following Lagrangian:
\begin{equation}
    L_\text{II}^{(N)} = \vec{x}_{n+1}\cdot\vec{x}_{n} 
    +\frac{\alpha}{2} \log\left(1- \vec{x}_{n}^{2}\right).
    \label{eq:LIIN}
\end{equation}

\begin{remark}
    We remark that the system \eqref{eq:vecdPII} was introduced for $N=2$ in 
    \cite{McLachlan1993} searching for explicit invariants for a general
    map in standard from \eqref{eq:2ndadd} with a given continuum limit.
    The $N$ degrees of freedom generalisation was suggested from radial symmetry,
    and proved to be a reduction of a discrete analogue of the Garnier
    system in \cite{Suris1994Garnier}. These systems were further 
    generalised in \cite{Suris1994Symmetric}, by considering symplectic 
    maps related to systems defined on symmetric spaces. For a general 
    review on these topics we refer to \cite[Chap. 25]{Suris2003book}.
    \label{rem:mclaclansuris}
\end{remark}

\subsection{Invariants of the system \eqref{eq:vecdPIIfull}}
In this section we will prove that the system \eqref{eq:vecdPIIfull}
is not Liouville integrable for all $N\in\N$,
but it possesses many invariants.

We start noticing that the system \eqref{eq:vecdPIIfull} is quasi-radially
symmetric, but not variational.
However, from Lemma \ref{lem:sl2general}, we have that the system
\eqref{eq:vecdPIIfull} possesses the $N\left( N-1 \right)\left( N-2 \right)/6$
invariants $K_{i,j,k}$ \eqref{eq:Kijk} and their Poisson-commuting
combinations \eqref{eq:casimir} 
(these can be defined even without a Poisson structure).

Moreover, we have that the system \eqref{eq:vecdPIIfull} admits the
coalgebra symmetry with respect to the $h_{6}$ algebra.
Indeed, we have the following associated dynamical system:
\begin{subequations}
    \begin{align}
        A_{+}^{(n+1)} &= 
        \frac{\alpha A_{+}^{(n)}+M^{(n)}}{1-B_{+}^{(n)}}-A_{-}^{(n)}, 
        \label{eq:vectdPIIfullh6evolAp}
        \\
        A_{-}^{(n+1)} &= 
        A_{+}^{(n)}, 
        \label{eq:vectdPIIfullh6evolAm}
        \\
        B_{+}^{(n+1)} &
        \begin{aligned}[t]
            &=B_{-}^{(n)}
        -\frac{\alpha^2+2 \alpha K^{(n)}+\alpha M^{(n)}+2 A_{-}^{(n)}}{1-B_{+}^{(n)}}
        \\
        &+\frac{\alpha^2+2 \alpha A_{+}^{(n)}+M^{(n)}}{\left(1-B_{+}^{(n)}\right)^2},
        \end{aligned}
        \label{eq:vectdPIIfullh6evolBp}
        \\
        B_{-}^{(n+1)} &= 
        B_{+}^{(n)},
        \label{eq:vectdPIIfullh6evolBm}
        \\
        K^{(n+1)} &= 
        \frac{\alpha B_{+}^{(n)}+A_{+}^{(n)}}{1-B_{+}^{(n)}} -K^{(n)}-M^{(n)}, 
        \label{eq:vectdPIIfullh6evolK}
        \\
        M^{(n+1)} &= 
        M^{(n)}.
        \label{eq:vectdPIIfullh6evolM}
    \end{align}
    \label{eq:vectdPIIfullh6evol}
\end{subequations}
The system is clearly closed in $h_{6}$.
By direct computation it is possible to show that the Lie--Poisson
bracket of $h_{6}$ is preserved, and finally that the central
element $M^{(n)}$ and the Casimir \eqref{eq:h6casimirthird} are preserved
by \eqref{eq:vectdPIIfullh6evol}.

Heuristically, the computation of the algebraic entropy of the
dynamical system \eqref{eq:vectdPIIfullh6evol} gives the degree
sequence \eqref{eq:seqmcmNcompl}.  So, we don't expect the system to be
integrable, but besides the central element $M^{(n)}$ and the Casimir
\eqref{eq:h6casimirthird}, we obtain the following invariant:
\begin{equation}
    \mathcal{B}^{(n)} =
    B_{+}^{(n)}B_{-}^{(n)}-B_{+}^{(n)}-B_{-}^{(n)} + \alpha \left( K^{(n)}+\frac{M^{(n)}}{2} \right)
    +A_{+}^{(n)}+A_{-}^{(n)}.
    \label{eq:vectdPIIfullB1}
\end{equation}
Note that the invariant $\mathcal{B}$ \eqref{eq:vectdPIIfullB1} is a direct
coalgebraic generalisation of the biquadratic \eqref{eq:mcmbiquadr}.

Thus, summing up these results, we get the following set of $N-1$ invariants:
\begin{equation}
    \mathcal{S}_{\text{II}} = 
    \left\{ \mathcal{B}^{(n)}, \mathcal{I}_{3}^{(n)},\dots,\mathcal{I}_{N}^{(n)} \right\}.
    \label{eq:PIIset}
\end{equation}
It is possible to prove by induction on the degrees of freedom that these
invariants are functionally independent.

The set \eqref{eq:PIIset} does not contain enough invariant to prove
integrability, especially because the system is lacking a Poisson
structure.  Considering the $2N-5$ invariants coming from Lemma
\ref{lem:qrad} we have a total of $2N-4$ independent invariants (this
can be proven again by induction).  However, we will see later that the
behaviour of the system is very regular.

\subsection{Quasi-maximal superintegrability of the $\vbeta=\vec{0}$ case \eqref{eq:vecdPII}}
The Liouville integrability of the $\vbeta=\vec{0}$ case follows
from the general case $\vbeta\neq\vec{0}$, noticing that all the invariants
are analytic in the neighbourhood of $\vbeta=\vec{0}$ and the system becomes radially symmetric.

However, here we give a short proof using the coalgebra symmetry $\Sl_{2}(\RR)$
and Theorem \ref{thm:radialquasiint}: if we find an additional commuting invariant
then the system becomes quasi-maximal superintegrable.
The associated dynamical system \eqref{eq:sl2evolutiongeneral} is:
\begin{equation}
    \begin{gathered}
        J_{+}^{(n+1)} =
        J_{-}^{(n)}
        -\frac{2\alpha J_{3}^{(n)}}{1-J_{+}^{(n)}}
        +\frac{\alpha^2 J_{+}^{(n)}}{\left(1-J_{+}^{(n)}\right)^2},
        \\
        J_{-}^{(n+1)} = J_{+}^{(n)},
        \quad
        J_{3}^{(n+1)} = -J_{3}^{(n)}
        +\frac{\alpha J_{+}^{(n)}}{1-J_{+}^{(n)}} .
    \end{gathered}
    \label{eq:sl2evolutionvecdPII}
\end{equation}

Heuristically, the computation of the algebraic entropy of the dynamical
system \eqref{eq:sl2evolutionvecdPII} gives the degree sequence
\eqref{eq:seqmcmNbeta0}, implying integrability.  If we search for
invariants of this system, besides the Casimir \eqref{eq:sl2casimir}, 
we obtain the following invariant:
\begin{equation}
    \mathcal{M}^{(n)} =
    J_{+}^{(n)} + J_{-}^{(n)}- J_{3}^{(n)} \left(\alpha+J_{3}^{(n)}\right).
    \label{eq:MdPII}
\end{equation}
Note that the invariant $\mathcal{M}^{(n)}$ \eqref{eq:MdPII} is a direct
coalgebraic generalisation of the biquadratic 
$\tilde{B}^{(n)}=B^{(n)}\left( x_{n},x_{n-1};\beta=0 \right)$,
with $B$ given by \eqref{eq:mcmbiquadr}.

So, from Theorem \ref{thm:radialquasiint} we build the set of invariants:
\begin{equation}
    \mathcal{P}_\text{II}\left( \vbeta=\vec{0} \right) = 
    \left\{\mathcal{M}^{(n)}, \mathcal{C}_{2}^{(n)},\dots,\mathcal{C}_{N}^{(n)}\right\}.
    \label{eq:invariantsdPII}
\end{equation}
Functional independence and involutivity in this set can be proved
by induction.
So, we proved that the system \eqref{eq:vecdPII} is
Liouville integrable, and then it is quasi-maximally superintegrable.

\subsection{Continuum limits}
It is known that the McMillan map \eqref{eq:mcm} has a cubic oscillator as
continuum limit, see for instance \cite{CresswellJoshi1999} for
the associated non-au\-ton\-o\-mous case.
For its $N$ degrees of freedom generalisation \eqref{eq:vecdPIIfull}
we have an analogous property.
The case $\vbeta=\vec{0}$ was considered in \cite{McLachlan1993,Suris1994Garnier}.

To be more precise, consider the following scaling:
\begin{equation}
    \vec{x}_{n} = h \vec{X}\left( t \right),
    \quad
    \alpha = 2 + h^{2} A,
    \quad
    \vbeta = h^{3} \vec{B},
    \quad 
    t = n h,
    \quad
    h \to 0.
    \label{eq:dPIIscaling}
\end{equation}
Substituting into the left hand side of equation \eqref{eq:vecdPIIfull} 
we get:
\begin{equation}
    \vec{x}_{n+1} + \vec{x}_{n-1}
    =
    h \left[ 2\vec{X}\left( t \right) + h^{2} \vec{\ddot{X}}\left( t \right) \right]
    + O\left( h^{4} \right).
    \label{eq:dPIIclimlhs}
\end{equation}
For the right hand side we have:
\begin{equation}
    \frac{\alpha\vec{x}_{n}+\vbeta}{1-\vec{x}_{n}^{2}}
    =
    2 h\vec{X}\left( t \right) + h^{3}
    \left[ \left( A +\vec{X}^{2} \right) \vec{X} + 2\vec{B}  \right]
    + O\left( h^{4} \right).
    %\begin{aligned}
    %    \frac{\alpha\vec{x}_{n}+\vbeta}{1-\vec{x}_{n}^{2}}
    %    &=
    %    \frac{\left( 2+A h^{2} \right) h \vec{X}\left( t \right) +h^{3}\vec{B}}{1-h \vec{X}^{2}\left( t \right)}
    %    \\
    %    &=
    %    2 h\vec{X}\left( t \right) + h^{3}
    %    \left[ \left( A +\vec{X}^{2} \right) \vec{X} + 2\vec{B}  \right]
    %    + O\left( h^{4} \right).
    %\end{aligned}
    \label{eq:dPIIclimrhs}
\end{equation}
So, comparing the two sides and balancing the terms in $h$ we obtain
the system:
%we have that 
%the first-order term is balanced, while the first non-trivial order 
%is $h^{3}$, whose balancing gives us the equation:
\begin{equation}
    \vec{\ddot{X}}
    = \left( A + \vec{X}^{2} \right) \vec{X} + 2\vec{B}.
    \label{eq:NDHirota}
\end{equation}
This vector system is a vector non-linear cubic oscillator, 
a natural generalisation of the one degree of freedom cubic oscillator.
Such a system possesses the following Lagrangian:
\begin{equation}
    \mathcal{L}_\text{II}=
    \frac{1}{2}\vec{\dot{X}}^{2}
    -\frac{1}{2}\left(A+\vec{X}^{2}  \right)\vec{X}^{2} -2\vec{B}\cdot\vec{X} .
    \label{eq:LIIcont}
\end{equation}
A trivial invariant is the energy (which coincides up to Legendre
transformation with the Hamiltonian):
\begin{equation}
    E = 
    \frac{1}{2}\vec{\dot{X}}^{2}
    +\frac{1}{2}\left(A+\vec{X}^{2}  \right)\vec{X}^{2} 
    +2\vec{B}\cdot\vec{X}.
    \label{eq:energydPII}
\end{equation}
In the integrable case \eqref{eq:vecdPII} we can also consider the
continuum limit in the Lagrangian \eqref{eq:LII}:
\begin{equation}
L_\text{II} = h^{4} \mathcal{L}_\text{II}\left( \vec{B}=\vec{0} \right) + O\left(h^{5}\right).
    \label{eq:LIIclim}
\end{equation}
Therefore, we can see clearly a drastic difference between the continuum
and the discrete case: the continuum case is always variational and
the variational structure itself yields ``for free'' the invariant
\eqref{eq:energydPII}.

Let us now discuss the invariants.  Under the scaling
\eqref{eq:dPIIscaling} the elements of the discrete angular momentum
become:
\begin{equation}
    L_{i,j}^{(n)} = h^{3} \ell_{i,j}+ O\left(h^{4}\right), 
    \quad
    \ell_{i,j} = -X_{i}\dot{X}_{j}+\dot{X}_{i}X_{j},
    \label{eq:lijcont}
\end{equation}
while the invariants $K_{i,j,k}$ become:
\begin{equation}
    K_{i,j,k}^{(n)} = h^{6} \kappa_{i,j,k}+ O\left(h^{7}\right), 
    \quad 
    \kappa_{i,j,k} = B_{i}\ell_{j,k}+B_{j}\ell_{k,i}+ B_{k}\ell_{i,j}.
    \label{eq:kijkcont}
\end{equation}
Following for instance \cite{Ballesteros_et_al2009}, we have that these
can be used to build the commuting set of invariants:
\begin{subequations}
    \begin{align}
        I_{m} &=
        \sum_{1\leq i<j<k\leq m} \kappa_{i,j,k}^{2}, 
        \quad
        m=3,\dots,N,
        \label{eq:Irmm}
        \\
        J_{m} &=
        \sum_{N-m+1\leq i<j<k\leq N} \kappa_{i,j,k}^{2}, 
        \quad 
        m=3,\dots,N.
        \label{eq:Jrmm}
    \end{align}
    \label{eq:kijlbuild}%
\end{subequations}
In the same way the invariant \eqref{eq:vectdPIIfullB1} is related
to the energy integral \eqref{eq:energydPII}:
\begin{equation}
    \mathcal{B}^{(n)} = h^{4} E+ O\left(h^{5}\right).
    \label{eq:Bnlimit}
\end{equation}
The set of invariants
$\mathcal{Q} = \left\{ E, I_{3},\dots,I_{N} \right\}$,
makes the system quasi-integrable.

In the quasi-maximally superintegrable case $\vbeta=\vec{0}$, we have that 
this property is preserved by the continuum limit, which clearly
correspond to put $\vec{B}=\vec{0}$.
To see this, just notice that $E\left( \vec{B}=\vec{0} \right)$ stays
an invariant, and that we can construct the left and right Casimirs
of $\Sl_{2}(\RR)$ from formula \eqref{eq:casimirsl2} as:
\begin{subequations}
    \begin{align}
        C_{m} &=
        \sum_{1\leq i<j\leq m} \ell_{i,j}^{2}, 
        \quad
        m=2,\dots,N,
        \label{eq:Crmm}
        \\
        D_{m} &=
        \sum_{N-m+1\leq i<j\leq N} \ell_{i,j}^{2}, 
        \quad 
        m=2,\dots,N.
        \label{eq:Drmm}
    \end{align}
    \label{eq:lijlbuild}%
\end{subequations}
So, the set of invariants
$\mathcal{Q}\left( \vec{B}=\vec{0} \right) = 
\left\{ E\left( \vec{B}=\vec{0} \right), C_{2},\dots,C_{N} \right\}$,
make the system Liouville integrable, and the existence
of the additional non-commuting invariants \eqref{eq:lijcont} 
make it quasi-maximally superintegrable.

To conclude this section we show some pictures highlighting the numerical
validation of the continuum limit we just presented.  In particular we
compare the orbits of the continuous case, ODE \eqref{eq:NDHirota}, with
its discrete counterpart, equation \eqref{eq:vecdPIIfull} after applying
the scaling \eqref{eq:dPIIscaling}. The two equations are solved in
four degrees of freedom with initial values $(\vec{x}_{0},\vec{x}_{-1})=
(0.1, 0.1, 0.1, 0.1 ,0.1,0,0,0.1)$, $\alpha = -2$, $h=0.5$ and $\vbeta=
(0.05,0.2,0.1,0.25)$.  In figure \ref{fig:orbits} the numerical
orbits are plotted.  In particular we note that both orbits look
very close to an integrable one.  Incidentally, this result exhibits
another limit of the ``orbit method'' to identify integrability, see
for instance \cite{Viallet2008IJGMP,GJTV_sanya}.  Finally, figures
\ref{fig:xdotbeta} and \ref{fig:norm} show the values of $B_+^{(n)} =
\boldsymbol{\beta}\cdot\mathbf{x}_n$ and $C_+^{(n)} =\mathbf{x}_n^2$  for
the continuous case and the discrete case.  These two figures highlight
a simple oscillatory behaviour, which in both cases mimic the behaviour
of the system with one degree of freedom.

\begin{figure}[ht]
    \centering
    \subfloat{\label{Fig_phaseplot_ode}\includegraphics[width=0.45\linewidth]{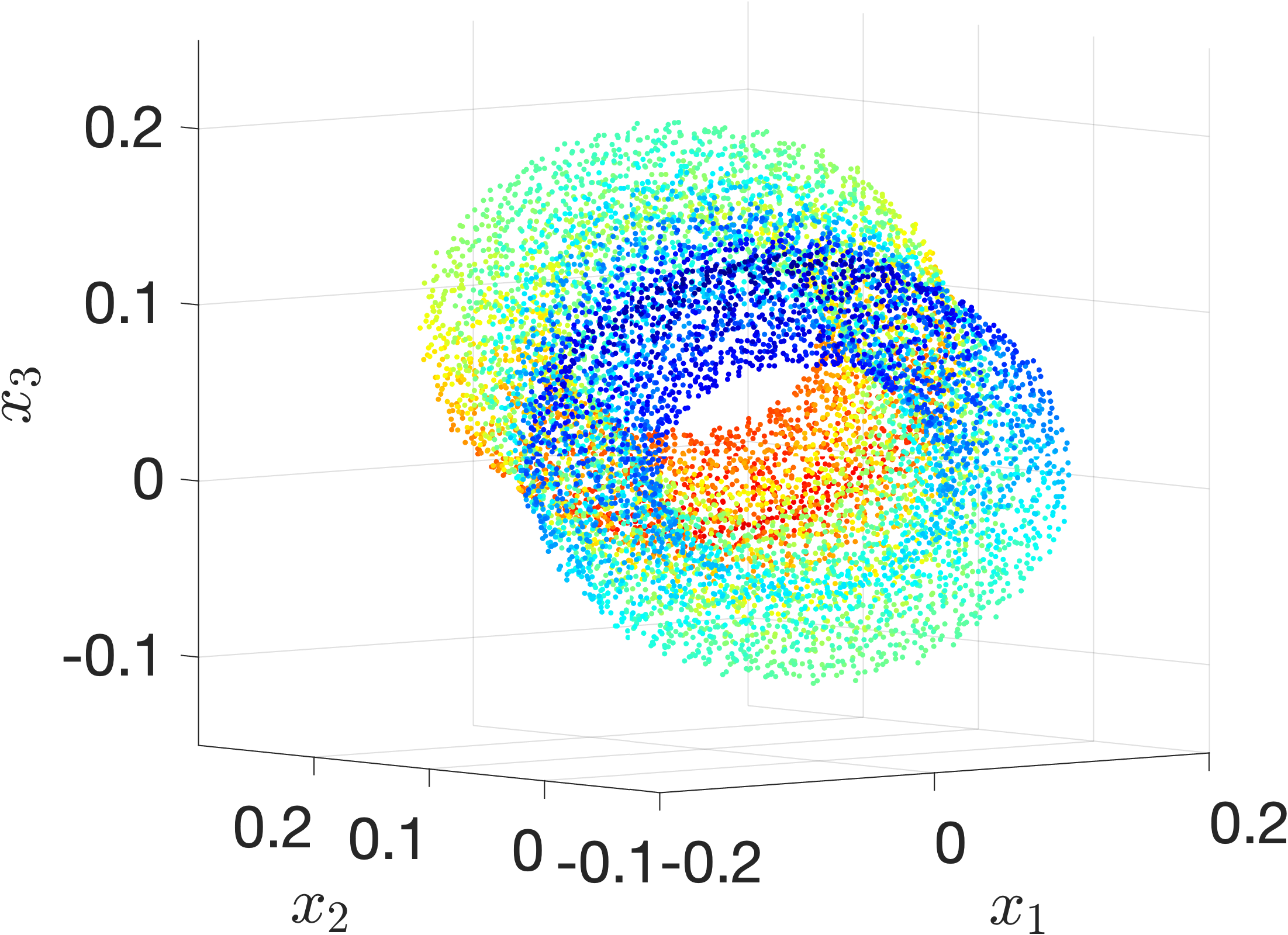}}
    \hspace{.5cm}
    \subfloat{\label{Fig_phaseplot_map}\includegraphics[width=0.45\linewidth]{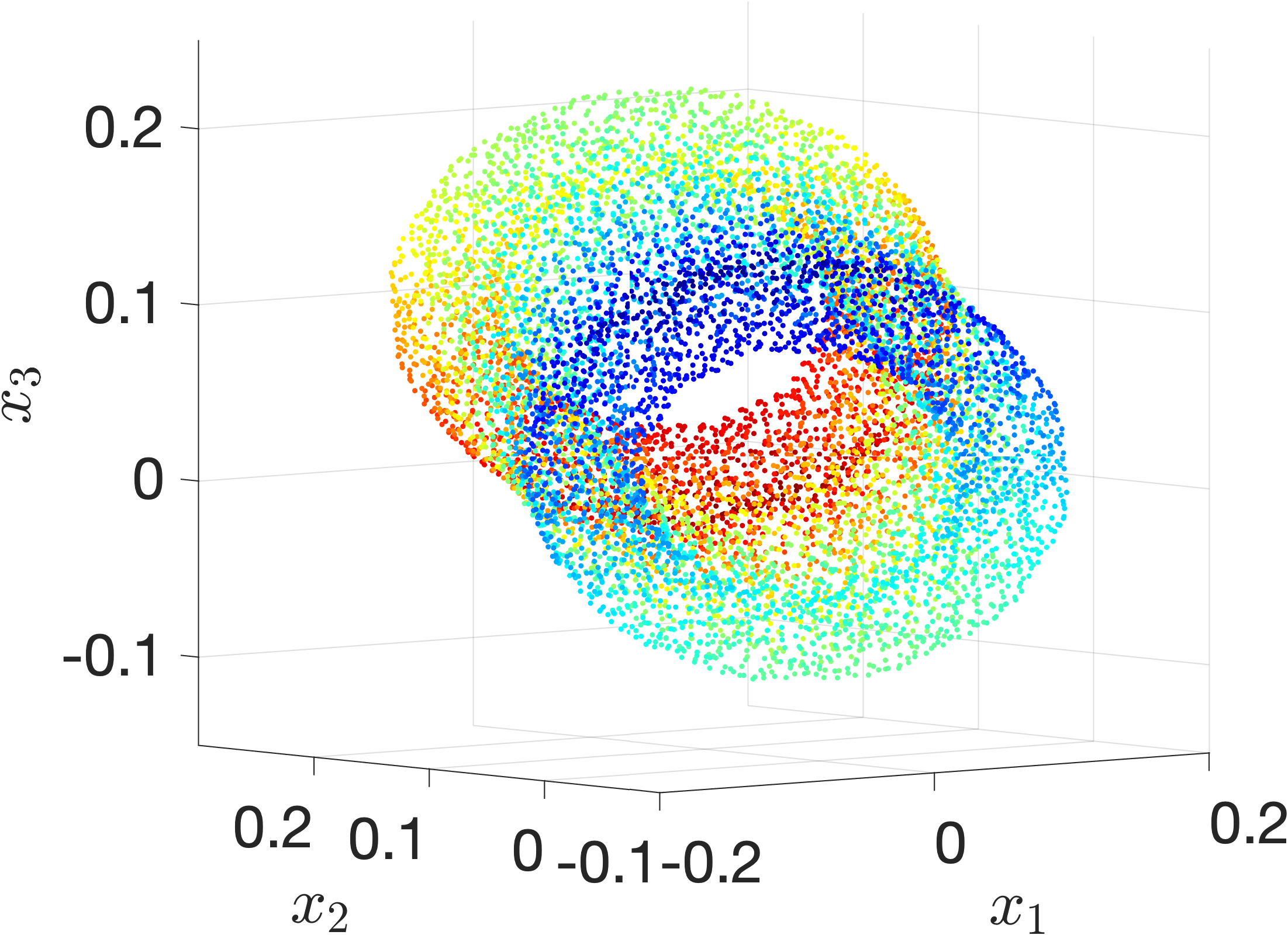}}
    \caption{Orbits of equation \eqref{eq:NDHirota} (left) 
    and \eqref{eq:vecdPIIfull} after applying the scaling 
    \eqref{eq:dPIIscaling} (right)}
    \label{fig:orbits}
\end{figure}

\begin{figure}[ht]
    \centering
    \subfloat{\label{Fig_xdotbeta_ode}\includegraphics[width=0.45\linewidth]{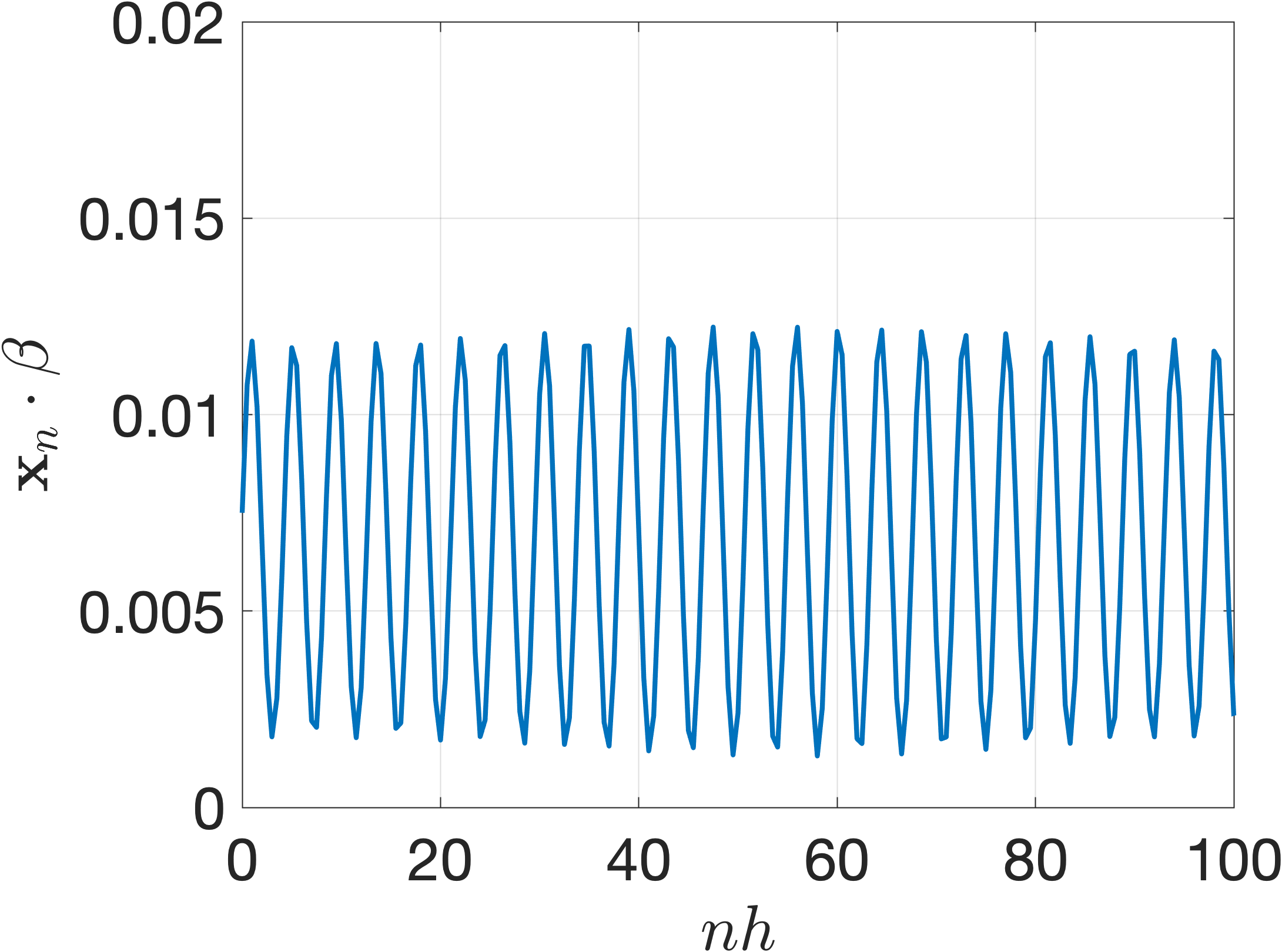}}
    \hspace{.5cm}
    \subfloat{\label{Fig_xdotbeta_map}\includegraphics[width=0.45\linewidth]{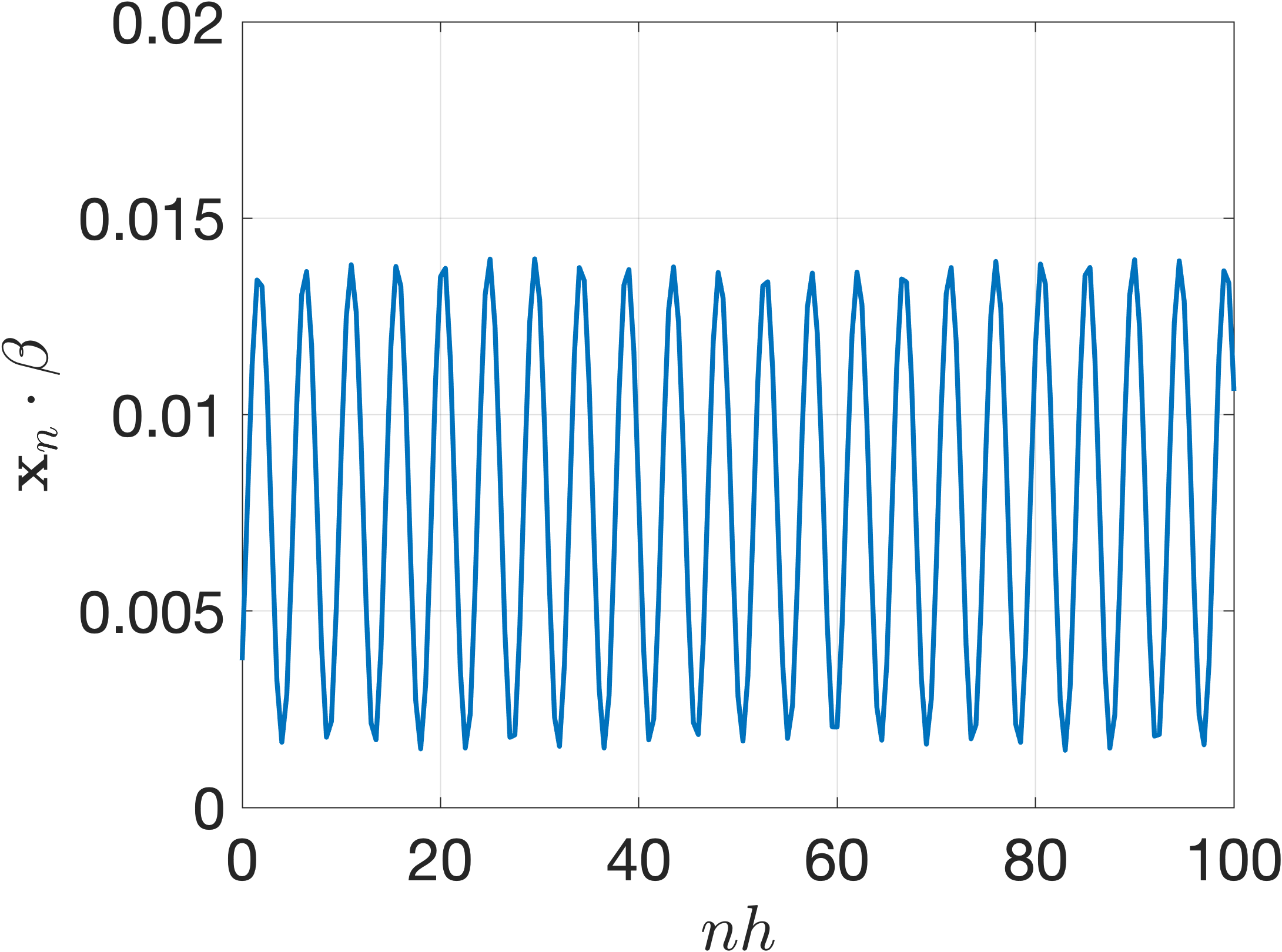}}
    \caption{$B_+^{(n)} = \boldsymbol{\beta}\cdot\mathbf{x}_n$ for the continuous case (left) and the discrete case (right).}
    \label{fig:xdotbeta}
\end{figure}

\begin{figure}[ht]
    \centering
    \subfloat{\label{Fig_norm_ode}\includegraphics[width=0.45\linewidth]{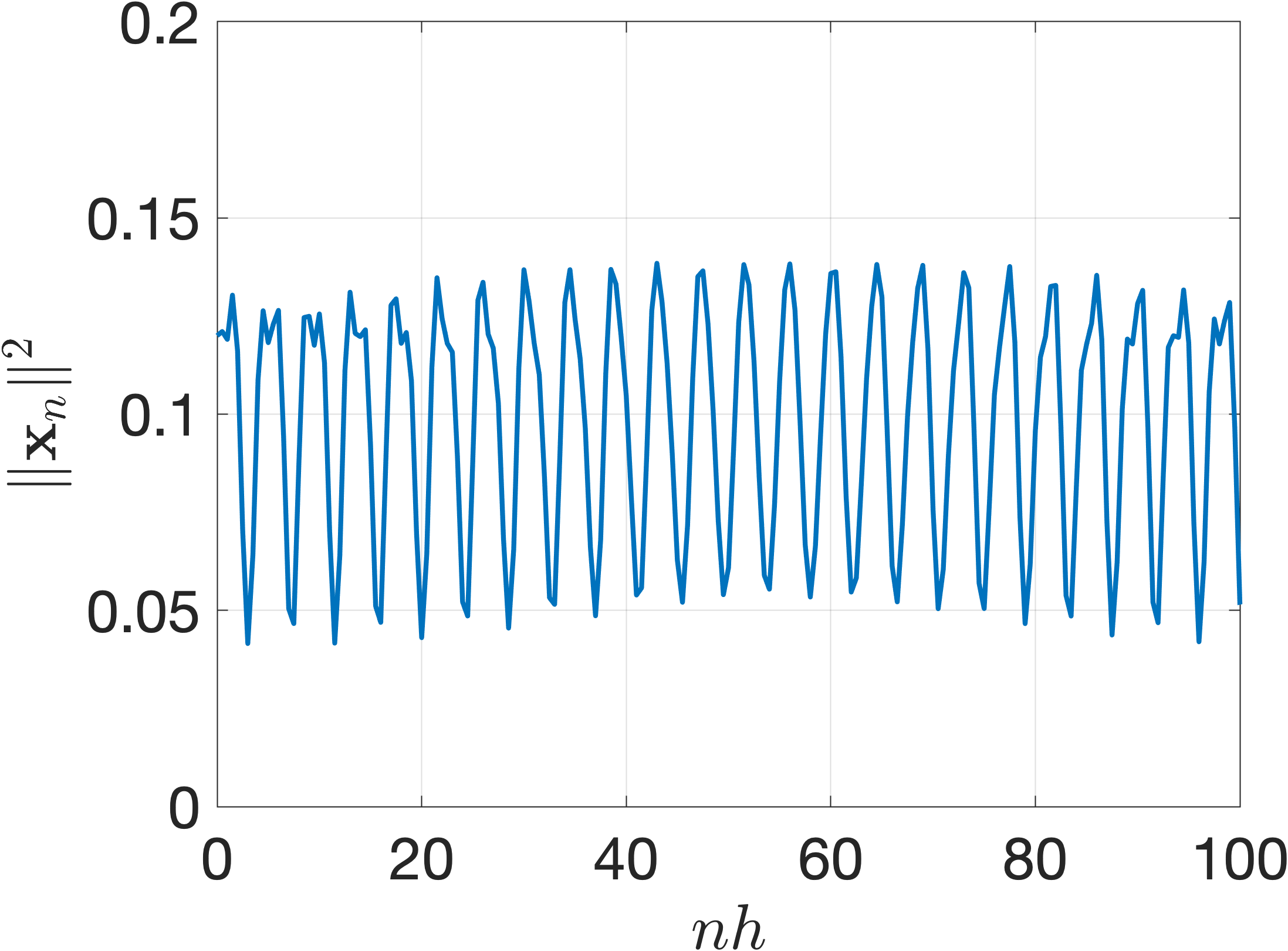}}
    \hspace{.5cm}
    \subfloat{\label{Fig_norm_map}\includegraphics[width=0.45\linewidth]{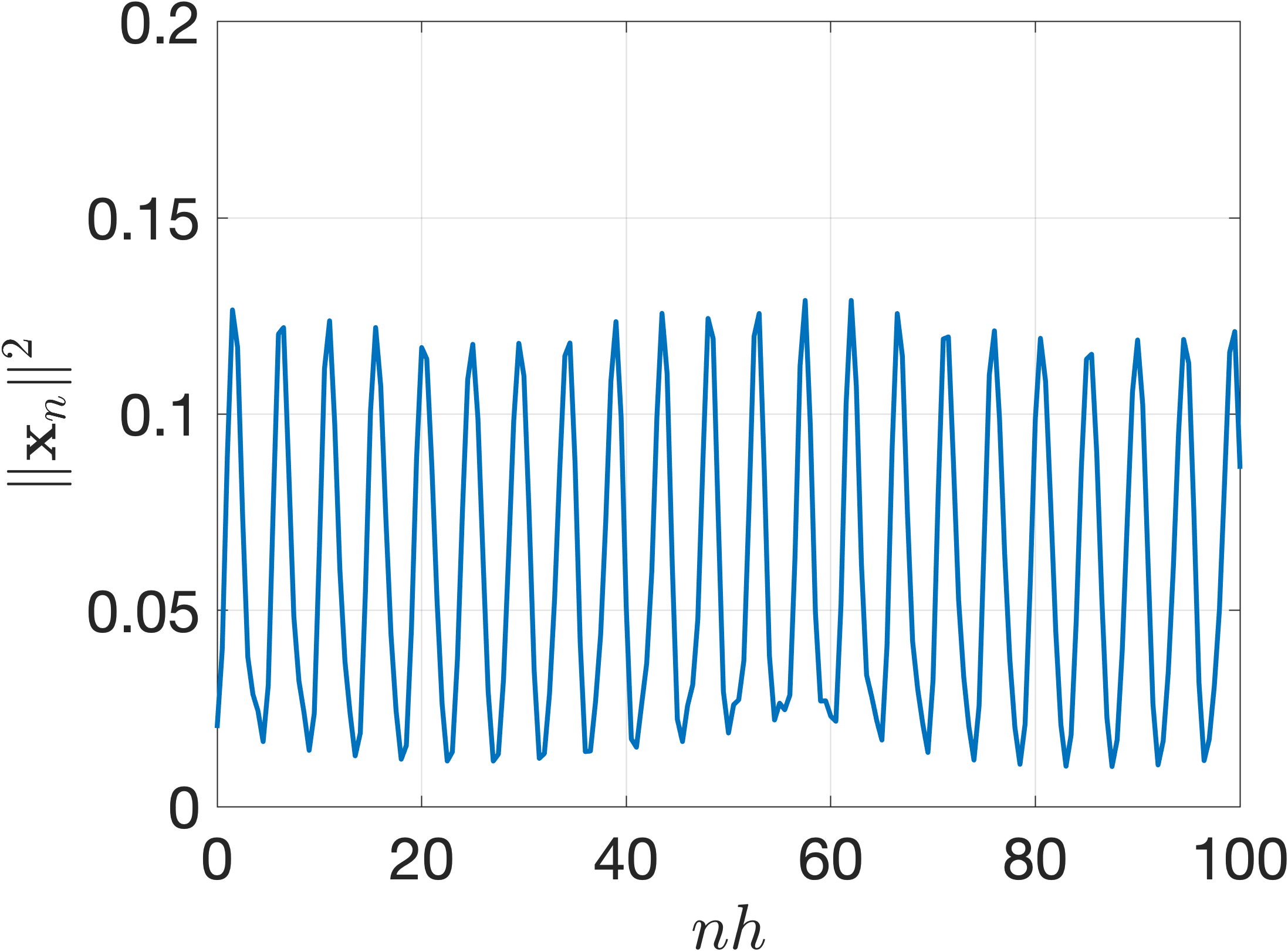}}
    \caption{$C_+^{(n)} = \mathbf{x}_n^2$ for the continuous case (left) and the discrete case (right).}
    \label{fig:norm}
\end{figure}

\section{Conclusions}
\label{sec:concl}

In this paper we presented a discrete analogue of the coalgebra
approach to generate systems in $N$ degrees of freedom as extensions
of integrable systems in one degree of freedom.  We gave two general
results on the integrability properties of two classes of systems of
second order difference equation in standard form.  Namely, we proved
in Theorem \ref{thm:radialquasiint} that all radially-symmetric
systems \eqref{eq:radsym} are quasi-integrable, and in Theorem
\ref{thm:qradialint} that all quasi-radially symmetric systems are PLN
maps of rank $N-2$. Up to our knowledge, this is the first time that
quasi-integrable discrete systems are produced. Then we considered the
following explicit examples of systems in $N$ degrees of freedom:
\begin{enumerate}
    \item A maximally superintegrable radially-symmetric linear 
        equation \eqref{eq:linear} in $N$ degrees of freedom.
    \item A superintegrable quasi-radially-symmetric linear equation 
        \eqref{eq:linearqr} in $N$ degrees of freedom.
    \item A quasi-integrable deformation of the previous system 
        \eqref{eq:nlinearqr}.
    \item Two different $N$ degrees of freedom generalisations
        of the autonomous discrete Painlev\'e I equation \eqref{eq:vecdPI}.
    \item A non-integrable $N$ degrees of freedom generalisation
        of the McMillan map \eqref{eq:vecdPIIfull} with many invariants.
    \item A quasi-maximally superintegrable radially-symmetric $N$ 
        degrees of freedom generalisation of a special case of the 
        McMillan map \eqref{eq:vecdPII}.
\end{enumerate}
In particular, the system \eqref{eq:nlinearqr} remarkably shows
the entropy gap: as soon as $\varepsilon\neq 0$ the $N$th invariant is
missing and the algebraic entropy becomes positive.
Despite this the (real) orbits of the system are very regular.
Furthermore, we remark that, while some of the systems we considered
were known in the literature, up to our knowledge the vector
generalisation of the autonomous discrete Painlev\'e I equation
\eqref{eq:vecdPIa} is new.

As usual in discrete systems theory, the discrete case is more complicated
than its continuum counterpart: while in the continuum case an $N$
degrees of freedom system is built on top of the $N$th coproduct of a
function, the Hamiltonian, in the discrete setting we define a Poisson
map to admit a coalgebra symmetry when the evolution of the generators
of the algebra is closed in the algebra and preserves the Casimir of the
algebra itself.  The last requirement is fundamental since it allows us
to prove the existence of the invariants \eqref{eq:Ctotg}, and then use
the construction to generate systems with $N$ degrees of freedom with
a given number of invariants.

The three conditions of Definition \ref{def:dcoalgebrasymmety}
can be used to build systems admitting the coalgebra symmetry out
of general ones.
In particular, we highlight with an example that the Casimir condition
can greatly help us.
For sake of simplicity consider a variational difference equation
of standard form \eqref{eq:2ndadd} for $N=2$ with $\vx_{n} = (x_{n},y_{n})^{T}$
and the algebra $\Sl_{2}(\RR)$.
Computing the evolution of the generators of $\Sl_{2}(\RR)$ we obtain:
\begin{subequations}
    \begin{align}
        J_{+}^{(n+1)} &= 
        \left[ x_{n-1}-\pdv{V}{x_{n}}(x_{n},y_{n}) \right]^{2}+
        \left[ y_{n-1}-\pdv{V}{y_{n}}(x_{n},y_{n}) \right]^{2},
        \\
        J_{-}^{(n+1)} &= J_{+}^{(n)},
        \\
        J_{3}^{(n+1)} &= 
        x_{n}\left[x_{n-1}-\pdv{V}{x_{n}}(x_{n},y_{n}) \right]^{2}+
        y_{n}\left[y_{n-1}-\pdv{V}{y_{n}}(x_{n},y_{n}) \right]^{2}.
    \end{align}
    \label{eq:algsl2clos}
\end{subequations}
We have that the commutation relations of the $\Sl_{2}(\RR)$ are preserved
for all $V$, while it is not trivial to understand if the right hand side
of the expressions in \eqref{eq:algsl2clos} are $\Sl_{2}(\RR)$.
However, it is pretty simple to check if the Casimir \eqref{eq:sl2casimir}
is preserved.
Computing the difference between $C^{(n+1)}$ and $C^{(n)}$ we obtain the
simple expression:
\begin{equation}
    \left(x_{n} \pdv{V}{y_{n}} -y_{n}\pdv{V}{x_{n}}  \right) 
    \left[x_{n} y_{n-1} -x_{n-1} y_{n} 
    -\frac{1}{2}\left(x_{n} \pdv{V}{y_{n}} -y_{n}\pdv{V}{x_{n}}  \right) \right]
    = 0.
    \label{eq:Veqsl2N2}
\end{equation}
Since the second factor cannot be zero because $V=V(x_{n},y_{n})$, the first
factor gives us a linear PDE for $V$. Solving it we obtain
$V=V\bigl(\sqrt{x_{n}^{2}+y_{n}^{2}}\bigr)$.
With such value of $V$ we have that the right hand side of the system 
\eqref{eq:algsl2clos} lies in $\Sl_{2}(\RR)$.
This reasoning can be extended to $N>2$, and we obtain that the only
variational difference equation of standard form \eqref{eq:2ndadd} admitting
the $\Sl_{2}(\RR)$ coalgebra are exactly the radial difference equations
\eqref{eq:radsym}.
Note that in this reasoning the variational structure is not restrictive
because we are interested in studying integrability.

We remark that in all the examples we presented the associated dynamical
system on the generators of the algebra \eqref{eq:closurerelation} is a
fundamental tool in studying the integrability of the original difference
equation.  In this sense the system \eqref{eq:closurerelation} plays a
role even more fundamental than the $N$ degrees of freedom Hamiltonian
which gives only one additional invariant.  This is particularly evident
in the case of equation \eqref{eq:vecdPIa} where, due to the presence
of the two-photon $h_{6}$ coalgebra, one should have constructed the
second invariant with other methods, see \cite{BallesterosHerranz2001}.
In particular, our examples seem to suggest that the integrability
properties of an underlying Poisson map are completely governed by those
of the associated dynamical system on the generators of the algebra
\eqref{eq:closurerelation}.  To be more precise, we make this statement
rigorous in the following conjecture:
\begin{conjecture*}
    A Poisson map $T\colon n \mapsto n+1$
    admitting a coalgebra symmetry $(\Alg,\Delta)$,
    is Liouville integrable \emph{if and only if} the evolution of
    the generators \eqref{eq:closurerelation} is Poisson--Liouville integrable.
\end{conjecture*}

This conjecture can be used as a guiding criterion to find more integrable
cases, starting from instance from a given coalgebra structure. As
we mentioned in the introduction, integrable discrete systems in one
degree of freedom are almost completely understood in terms of QRT
mappings \cite{QRT1988,QRT1989}.  QRT mappings have been classified
in nine canonical forms in \cite{Ramanietal2002}.  The additive form
\eqref{eq:mst1d} is the first of these nine canonical forms.  We plan to
address to the problem of finding the $N$ degrees of freedom version of
these maps admitting some notable coalgebra symmetry, like the $\Sl_2(\RR)$
algebra or the $h_{6}$ algebra.

We also note that the production of many examples of discrete integrable
systems in $N$ degrees of freedom could help to understand the geometric
mechanism behind integrability for systems with many degrees of freedom.
Indeed, while for one degree of freedom discrete integrable systems
almost all properties can be explained in terms of involution on elliptic
curves and fibrations \cite{Sakai2001,Duistermaat2011book,Tsuda2004}, it
is known that integrable systems in more degrees of freedom are not always
related to elliptic fibrations, see \cite{GJTV_sanya,JoshiViallet2017,GJTV_class}.

Finally, we also plan to address to the problem of finding non-autonomous
versions of the coalgebraic integrable systems, e.g. with the singularity
confinement method \cite{Grammaticosetal1991}, and prove rigorously their
growth properties with the approach from \cite{Viallet2015}.

\section*{Acknowledgements}

We thank Prof. Yu. B. Suris for bringing to our attention his
interesting papers related to this work.

GG has been supported by Fondo Sociale Europeo del Friuli Venezia Giulia,
Programma operativo regionale 2014--2020 FP195673001 (A/Prof. T. Grava
and A/Prof. D. Guzzetti).

DL was supported by Australian Research Council Discovery Project DP190101529
(A/Prof. Y.-Z. Zhang).

BKT was supported by the European Unions Horizon 2020 research and innovation
programme under the Marie Sklodowska-Curie grant agreement (No. 691070)
    
\appendix

\section{Algorithm to find integrals of discrete systems}
\label{app:findinv}

In this appendix we recall briefly a method for
finding invariants of birational maps presented first
in \cite{FalquiViallet1993} and recently reprised in
\cite{CelledoniEvripidouMcLareOwewnQuispelTapleyvanderKamp2019},
where such a result was interpreted in terms of \emph{discrete Darboux
polynomials}.

In the case when a $M$-dimensional difference equation \eqref{eq:zeq}
is rational we can transform its map form into a projective map by
homogenising the variables.  To use all the advantages of projective and
algebraic geometry we usually consider the projective space to be defined
on the complex field, that is we consider the complex projective space
of dimension $M$, $\Pj^{M}$, with coordinates $[X_{1}:\dots:X_{M+1}]$.
In such a case, we denote the map by $\varphi\colon \Pj^{M}\to\Pj^{M}$.
In the case when the map possess an inverse $\psi\colon \Pj^{M}\to\Pj^{M}$
which is a rational map too then the map is said to be \emph{birational}
\cite{Shafarevich1994}.  Due to birationality the following relations
hold:
\begin{equation}
    \psi \circ \varphi= \kappa \Id,
    \quad
    \varphi \circ \psi= \lambda \Id,
    \quad
    \kappa,\lambda\in\Cp_{h}\left[ x_{1},\cdots,x_{n+1} \right].
    \label{eq:kappadef}
\end{equation}
The polynomials $\kappa$ and $\lambda$ admit a possibly trivial 
factorisation  of the form:
\begin{equation}
    \kappa = \prod_{i=1}^{K_{\kappa}} \kappa_{i}^{d_{\kappa,i}},
    \quad
    \lambda = \prod_{i=1}^{K_{\lambda}} {\lambda}_{i}^{d_{\lambda,i}}.
    \label{eq:kappafact}
\end{equation}
The map $\varphi$ is ill-defined on the singular locus $V_{\varphi}=
\left\{ \kappa=0 \right\}$, while the map $\varphi$ is ill-defined on
the singular locus $V_{\psi}= \left\{ \lambda=0 \right\}$.  The singular
loci form an algebraic variety of codimension one and measure zero.

In this picture an invariant is a homogeneous function $I=I(X_{1},\dots,X_{M+1})$
such that the pullback 
\begin{equation}
    \varphi^{*}(I)(X_{1},\dots,X_{M+1}):=I(\varphi(X_{1},\dots,X_{M+1})),
    \label{eq:pullback}
\end{equation}
satisfies $\varphi^{*}(I) = I$.
Now, if the invariant is a ratio of homogeneous polynomials,
that is $I\in\Cp_{h}(X_{1},\dots,X_{M+1})$, we can write
$I=P/Q$, with $P,Q\in\Cp_{h}[X_{1},\dots,X_{M+1}]$.
Equation \eqref{eq:pullback} implies:
\begin{equation}
    \varphi^{*}(P) = a P \quad \text{and} \quad \varphi^{*}(Q) = a Q
    \label{eq:covariance}
\end{equation}
for some polynomial factor $a\in\Cp_{h}[X_{1},\dots,X_{M+1}]$. 
That is the polynomials $P,Q$ are \emph{covariant}.

To find covariant polynomials we use the fact that the polynomial $a$ must
be composed by the factors of the polynomial $\kappa$, see \cite[Lemma
4.1]{FalquiViallet1993}. So, we can search for invariants imposing the
form of $P$, then searching for the appropriate cofactors, building them
from the factorisation \eqref{eq:kappafact}.  We get an invariant when we
obtain more than one solution for the \emph{same $a$}. By taking ratios
of the solutions we obtain the invariants.

The disadvantage of this algorithm is that it is not bounded as we don't
know a priori the degree of $P$.  However, in practice this approach
is quite useful for the explicit computation of the invariants, since
the conditions in \eqref{eq:covariance} are \emph{linear}, even though
their number can become huge as $\deg (P)$ and $M$ grow.

\section{Algebraic entropy}

An integrability criterion unique to birational systems with
discrete degrees of freedom is \emph{low growth condition}
\cite{Veselov1992,FalquiViallet1993,BellonViallet1999}. To be specific,
we state the following criterion of integrability:
\begin{definition}[Algebraic entropy \cite{BellonViallet1999}]
    Consider an $M$-dimensional difference equation \eqref{eq:zeq}.
    If the associated projective map $\varphi\colon \Pj^{M}\to\Pj^{M}$
    is birational, then we say that equation \eqref{eq:zeq} is 
    \emph{integrable in the sense of the algebraic entropy}
    if the following limit
    \begin{equation}
        S_{\varphi} = \lim_{k\to\infty}\frac{1}{k}\log \deg \varphi^{k},
        \label{eq:algentdef}
    \end{equation}
    \label{def:algent}
    called the \emph{algebraic entropy} is zero for every initial
    condition.
\end{definition}

Algebraic entropy is an \emph{invariant} of birational maps, meaning
that its value is unchanged up to birational equivalence.
Practically algebraic entropy is a measure of the \emph{complexity} of
a map, analogous to the one introduced by Arnol'd \cite{Arnold1990}
for diffeomorphisms.
In this sense growth is given by computing the number of intersections of the
successive images of a straight line with a generic hyperplane in
complex projective space \cite{Veselov1992}.

In principle, the definition of algebraic entropy in equation
\eqref{eq:algentdef} requires us to compute all the iterates of
a birational map $\varphi$ and take the limit as $k\to\infty$.
However, in the majority of applications, the asymptotic behaviour
of the sequence of degrees can be inferred by using generating
functions \cite{Lando2003}:
\begin{equation}
    g\left( z \right) = \sum_{n=0}^{\infty} d_{k}z^{k},
    \quad
    d_{k}=\deg\varphi^{k}.
    \label{eq:genfunc}
\end{equation}
A generating function is a predictive tool which can be used to test the
successive members of a finite sequence.  It follows that the algebraic
entropy is given by the logarithm of the smallest pole of the generating
function, see \cite{GubbiottiASIDE16,GrammaticosHalburdRamaniViallet2009}.
A birational map (or its avatar difference equation) will then be
integrable if all the poles of the generating function lie on the
unit circle.

\printbibliography

%\bibliography{bibliography}% common bib file

\end{document}